\documentclass{amsart}
\usepackage{graphicx}
\usepackage{graphicx}
\usepackage{amssymb}
\usepackage{amsfonts}
\usepackage{amsmath}
\usepackage{amsthm}
\newtheorem{thm}{Theorem}[section]

\newtheorem{lem}[thm]{Lemma}
\newtheorem{rem}[thm]{Remark}

\newtheorem{prop}[thm]{Proposition}

\numberwithin{equation}{section}

\newcommand{\cR}{\mathcal{R}}

\begin{document}
%\begin{frontmatter}
\title{Risk Minimization for Game Options in Markets Imposing Minimal Transaction Costs}
 %\runtitle{{Numerical schemes for $G$--Expectations}

 \author{Yan Dolinsky and Yuri Kifer \\
 Hebrew University of Jerusalem
}%

\address{
 Department of Statistics, Hebrew University of Jerusalem, Israel\\
 {e.mail: yan.yolinsky@mail.huji.ac.il}}
 \address{
 Department of Mathematics, Hebrew University of Jerusalem, Israel\\
 {e.mail: kifer@math.huji.ac.il}}
 %\affiliation{ETH Zurich}
 %\affiliation{ETH Zurich}
%\runauthor{Y. Dolinsky}

\date{\today}

\begin{abstract}
We study partial hedging for game options in markets with transaction costs
bounded from below. More precisely, we assume that the investor's transaction costs for each trade
are the maximum between proportional transaction costs and a fixed transaction costs. We prove that in the continuous time Black--Scholes (BS) model,
there exists a trading strategy which minimizes the shortfall risk. Furthermore, we use binomial models in order to provide numerical schemes for the calculation
of the shortfall risk and the corresponding optimal portfolio in the BS model.
\end{abstract}

\maketitle

\markboth{Y.Dolinsky and Y.Kifer}{Game Options in Markets Imposing Minimal Transaction Costs}
\renewcommand{\theequation}{\arabic{section}.\arabic{equation}}
\pagenumbering{arabic}

\section{Introduction}\label{sec:1}\setcounter{equation}{0}
A game contingent claim (GCC) or game option, which was introduced in
\cite{Ki1}, is defined
as a contract between the seller and the buyer of the option such
that both have the right to exercise it at any time up to a
maturity date (horizon) $T$.
If the buyer exercises the contract
at time $t$ then he receives the payment $Y_t$, but if the seller
exercises (cancels) the contract before the buyer then the latter
receives $X_t$. The difference $\Delta_t=X_t-Y_t$ is the penalty
which the seller pays to the buyer for the contract cancellation.
In short, if the seller will exercise at a stopping time
$\sigma\leq{T}$ and the buyer at a stopping time $\tau\leq{T}$
then the former pays to the latter the amount $H(\sigma,\tau)$
where
$H(\sigma,\tau)=X_{\sigma}\mathbb{I}_{\sigma<\tau}+Y_{\tau}\mathbb{I}_{\tau\leq{\sigma}}$
and we set $\mathbb{I}_{Q}=1$ if an event $Q$ occurs and
$\mathbb{I}_{Q}=0$ if not.

A hedge (for the seller) against a GCC is defined as a pair
$(\pi,\sigma)$ that consists of a self financing strategy $\pi$
and a stopping time $\sigma$ which is the cancellation time
for the seller whom we also call an investor. For more details
 on game options see \cite{Ki2}.

In this paper we study hedging with transaction costs of the following form. If the investor
makes a small trade then he pays a fixed transaction costs,
and if the investor makes a large trade he pays proportional transaction costs. Formally,
for buying (or selling) $\beta\neq 0$ stocks
the transaction costs are given by $\max(\delta,\mu|\beta| S)$
where $\delta>0$, $0<\mu<1$ are constants and $S$ is the stock price
at the moment of the trade.
The investor's total transaction costs should be finite, hence in our setup
the investor can trade only a finite (random) number of times.
Although this type of transaction costs is very natural and widespread it was not
 studied much so far.

In \cite{D1} it was proved that super--replication
of game options under proportional transaction costs is expensive and leads to trivial (buy--and--hold) strategies.
Since the friction in our setup is larger than the friction
in the proportional costs setup, then similar results hold true for our case as well.
Therefore, with the presence of transaction costs,
it is reasonable to assume that the seller's (investor's) initial
capital is less than the superhedging price, and so hedging with risk comes into the picture.
We deal with a certain type of risk called the
shortfall risk, which is the maximal expectation
with respect to the buyer exercise times of the shortfall. For the definition of
the shortfall risk measure for game options see \cite{DK1} and \cite{DK2}.

There are several papers which study shortfall minimization with friction (see for instance,
\cite{D}, \cite{G1}, \cite{G2}, \cite{Ka1} and \cite{Ka2}). All of these papers considered the proportional
transaction costs setup (for European and American options)
for which they proved an existence of
an optimal hedge.

In real market conditions transaction costs
generally contain a fixed component, i.e. the transaction costs are bounded from below by a positive constant.
Many authors considered utility maximization under transaction costs with a fixed component. For details see
 (\cite{AJS}, \cite{EH}, \cite{Ko}, \cite{LMW}, \cite{MP} and \cite{OP}).
 However for partial hedging of derivative securities
this setup was not studied before.

In this paper we consider a game option in the BS model with continuous path dependent payoffs.
Our first result says that for our type of frictions there is an optimal hedge. In general, the problem of the existence
of an optimal hedge for the shortfall risk measure in a game options setup is much more complicated than for European and American options.
The reason is that for game options the shortfall risk measure fails to be a convex functional of the portfolio strategy,
and so the compactness principle
which relies on the Komlos lemma can not
applied here. This is the principle which applied for European and
American options in all of the mentioned above papers.
For the type of transaction costs considered here convexity arguments do not work for any
type of options, and so our result is new even for American options (as a special case of game options) though in this case the proof is simpler.
For game options in a setup without friction or with
friction smaller than we consider in this paper,
the existence of a shortfall minimization hedge is an open question.

Our approach is to establish the continuity of the shortfall risk function
and to reduce the optimization problem to a Dynkin game with continuous payoffs.
Then we apply Dynkin games theory in Brownian setup and use the fact
that the set of all permissible trades is compact.

Next, we deal with the computational aspect of shortfall risk minimization. We employ
an appropriate sequence of binomial models in order to
approximate the shortfall risk and to construct ”almost” optimal portfolios in the
BS model. For Lipschitz continuous (may be path dependent) payoffs we obtain two sided error estimates
for the binomial approximations.
So far, shortfall risk approximations for game options were considered only
in the frictionless setup (see \cite{DK2}), where we obtained only one sided error estimates.
In general, in the presence of friction which is no less than proportional transaction costs, we can find a uniform bound for
the growth of \textit{admissible} portfolios. This bound is
essential for establishing the error estimates.

For game options in binomial models the shortfall risk and the corresponding optimal
portfolio can be calculated by a dynamical programming algorithm.
Thus, these approximation theorems provide an efficient tool for numerical calculation
of the shortfall risk and the corresponding optimal hedge in the BS model.

Main results of this paper are formulated in the next section.
In Section \ref{sec4} we prove the existence of an optimal hedge (Theorem \ref{thm2.1}).
In Section \ref{sec2+} we prove the approximation results (Theorem \ref{thm2.2}). The
proof of both Theorems \ref{thm2.1} and \ref{thm2.2} rely on some regularity properties
of the shortfall risk whose proofs we postpone till Section \ref{sec3}.

\section{Preliminaries and Main Results}
\label{sec:2}
\setcounter{equation}{0}

Consider a complete probability space
($\Omega, \mathcal{F}, \mathbb P)$ together with a standard
one--dimensional Brownian motion
\{$W_t\}_{t=0}^\infty$, and the filtration
$\mathcal{F}_t=\sigma{\{W_s|s\leq{t}\}}$ completed by the null sets.
Our BS financial market consists of a safe asset $B$ used
as numeraire, hence $B\equiv 1$,
and of a risky asset $S$ whose value at time $t$ is given by
\begin{equation}\label{2.new}
S^{(s)}_t=s\exp(\kappa W_t+(\vartheta-\kappa^2/2) t), \ \ s>0, \ \ t\geq 0
\end{equation}
where $\kappa>0$ is called volatility and $\vartheta\in\mathbb{R}$ is
another constant. It is well known that for the BS model there exists a unique probability measure $\mathbb{Q}\sim\mathbb{P}$
such that the stock price process $S^{(s)}$ is a $\mathbb Q$ martingale. Using standard arguments it follows
that the restriction of
the probability measure $\mathbb{Q}$ to
the $\sigma$--algebra
$\mathcal{F}_t$ satisfies
\begin{equation}\label{2.martingale}
Z_t:=\frac{d\mathbb Q}{d\mathbb P}|\mathcal{F}_t
=\exp\left(-\frac{\vartheta}{\kappa} W_t-\frac{1}{2}\left(\frac{\vartheta}{\kappa}\right)^2 t\right).
\end{equation}

Next, let $T<\infty$ and let $C[0,T]$ be the space of all continuous functions
$f:[0,T]\rightarrow\mathbb R$ equipped with the uniform topology. Denote by $C_{++}[0,T]\subset C[0,T]$
the subset of all strictly positive functions.
Let $F,G:C[0,T]\rightarrow C[0,T]$
be continuous progressively measurable functions which means
that for any $t\in [0,T]$ and $x,y\in C[0,T]$,
$G(x)_{[0,t]}=G(y)_{[0,t]}$ and
$F(x)_{[0,t]}=F(y)_{[0,t]}$ if
$x_{[0,t]}=y_{[0,t]}.$
We assume that $F\leq G$ and there exist constants $C,p>0$ for which
\begin{equation}\label{2.3}
||F(x)||+||G(x)||\leq C(1+||x||^p)
\end{equation}
where $||\cdot||$ denotes the sup norm on the space $C[0,T]$.
Consider a game option with maturity date $T$ and continuous payoffs which are given by
$$Y^{(s)}_t=[F(S^{(s)})](t)\leq [G(S^{(s)})](t)=X^{(s)}_t.$$

In our model, purchase and sale of
the risky asset are subject to transactions costs which are the maximum
of a constant fee and a proportional transaction cost. Namely,
if the investor
buys (or sells) $\beta$ stocks
then his transaction costs are given by
$$g(\beta,S):=\max(\delta,\mu|\beta| S)\mathbb I_{\beta\neq 0}$$
where $\delta>0$, $0<\mu<1$ are constants and $S$ is the stock price
at the moment of trade.
Presence of this minimal transaction cost yields that in order to avoid infinite transaction costs
portfolios can only be rebalanced finitely many (but a random number of) times.

Next, we define hedging and shortfall risk in the above setup.
A (self financing) trading strategy with an initial position $(z,y)$ is a triple $\pi=(z,y,\gamma)$
where $z$ is the cash value of the portfolio at the initial time, $y$
is the number of
stocks at this moment, and $\gamma={\{\gamma_t\}}_{t=0}^T$ is
an adapted, left continuous, pure jump process with finite (random) number of jumps and initial
value $\gamma_0=y$.
The random variable $\gamma_t$
denotes the number of shares in the portfolio $\pi$
at time $t$ before any change is made at this time (which
 is the reason why we assume that the process $\gamma$ is left continuous).
Observe that at time $0$ the investor has the value $z+g(y,s)-ys$ on his savings account.
Thus the portfolio (cash) value of a trading strategy $\pi$ at time $t$ is given by
\begin{equation}\label{2.4}
V^{\pi}_t=z+\int_{0}^t \gamma_u dS^{(s)}_u+g(y,s)-
g(\gamma_t,S^{(s)}_t)
-\sum_{u\in [0,t)}g(\gamma_{u+}-\gamma_u, S^{(s)}_u),
\end{equation}
where in the last sum there is only finitely many terms
which are not equal to zero.
A portfolio $\pi$ will be called \textit{admissible} if $V^\pi_t\geq 0$ for any $t$.
 A hedge consists of a trading strategy and a cancellation time. Thus, formally
 a hedge with initial position $(z,y)$ is a pair $(\pi,\sigma)$ such that
 $\pi$ is an \textit{admissible} portfolio and $\sigma\leq T$ is a stopping time (with respect to the Brownian filtration).
From (\ref{2.4}) it follows that for an \textit{admissible} portfolio $\pi$ the
stochastic process $V^\pi_t$, $t\geq 0$ is a super--martingale
with respect to the martingale measure $\mathbb Q$.
The set of all
hedges with initial position $(z,y)\in\mathbb{R}_{+}\times\mathbb{R}$ will be denoted by
$\mathcal{A}(T,s,z,y)$. The set of all hedges will be denoted by
$\mathcal{A}(T,s)$, where
$s$ is the initial stock price and
$T$ is the maturity date.

Next, we define the shortfall risk. Denote by $\mathcal{T}_{T}$ the set of all stopping times less or equal than $T$.
 For a hedge $(\pi,\sigma)$ the shortfall risk is defined by
\begin{equation*}\label{2.5}
\cR(T,s,\pi,\sigma)=\sup_{\tau\in\mathcal T_{T}} \mathbb{E}_{\mathbb P}\left(X^{(s)}_{\sigma}\mathbb{I}_{\sigma<\tau}+
Y^{(s)}_{\tau}\mathbb{I}_{\tau\leq\sigma}-V^\pi_{\sigma\wedge\tau}\right)^{+},
\end{equation*}
which is the maximal possible expectation with respect to the
probability measure $\mathbb P$ of the shortfall.
The shortfall risk for an initial position $(z,y)$ is given by
\begin{equation*}\label{2.6}
R(T,s,z,y)=\inf_{(\pi,\sigma)\in\mathcal{A}(T,s,z,y)}
\cR(T,s,\pi,\sigma).
\end{equation*}
The following theorem says that for a given initial position $(z,y)$ there exists a hedge which minimizes the shortfall risk.
\begin{thm}\label{thm2.1}
Let $(z,y)\in\mathbb{R}_{+}\times\mathbb R$ be an initial position.
There exists a hedge (may be not unique) $(\hat\pi,\hat\sigma)\in \mathcal A(T,s,z,y)$ such that
$\cR(T,s,\hat\pi,\hat\sigma)=R(T,s,z,y).$
\end{thm}
Next, we approximate the shortfall risk in the Black--Scholes model by a sequence of binomial models. In order to obtain error estimates we will
assume that the functions $F,G$
can be extended to the space $D[0,T]$ (of all $c\grave{a}dl\grave{a}g$ functions on the interval $[0,T]$)
and satisfy the following Lipschitz condition. There exists a constant $L$ such that
for any $0\leq t_1<t_2\leq T$ and $x,y \in  D[0,T]$,
\begin{eqnarray}\label{condition}
&||F(y)-F(x)||+||G(y)-G(x)||\leq L ||y-x|| \ \ \mbox{and}  \\
&|F(x)(t_2)-F(x)(t_1)|+|G(x)(t_2)-G(x)(t_1)|\leq \nonumber\\
&L \left((t_2-t_1)(1+||x||)+\sup_{t_1\leq u\leq t_2}|x(u)-x(t)|\right).
\nonumber
\end{eqnarray}
For any $n$ consider a
binomial model which consists
of a savings account $\equiv 1$,
and of a piecewise constant risky asset $\{S^{n,s}_t\}_{t=0}^T$ given by
\begin{equation*}\label{2.7}
S^{n,s}_t=s\exp\left(\kappa\sqrt{\frac{T}{n}}\sum_{i=1}^{[nt/T]}\xi_i\right), \ \ t\in [0,T]
\end{equation*}
where $[\cdot]$
denotes the integer part  and
$\xi_1,\xi_2,...$ are i.i.d. random variables taking values 1
and $-1$ with probabilities
$p^{(n)}=\left(\exp((\kappa-\frac{2\vartheta}{\kappa})\sqrt{\frac{T}{n}})+1\right)^{-1}$
and
$1-p^{(n)}=\left(\exp((\frac{2\vartheta}{\kappa}-\kappa)\sqrt{\frac{T}{n}})+1\right)^{-1}$,
respectively.
Let $\mathbb P_n$
be the corresponding probability measure and let $\mathcal{F}^{(n)}_t=\sigma\{S^{n,s}_u:u\leq t\}$ be
the filtration which is generated by $S^{n,s}$.

Denote by $\mathcal{A}^{(n)}(T,s,z,y)$ the set of all hedges with an initial position
$(z,y)$. The definition of a hedge is done in analogous way to the Black--Scholes model,
just replace the Brownian filtration by $\mathcal F^{(n)}$ and
$S^{(s)}$ by $S^{n,s}$ in (\ref{2.4}).

We introduce game options with the piecewise constant payoffs
$$Y^{n,s}_t=[F(S^{n,s})]([nt/T]T/n)\leq X^{n,s}_t=[G(S^{n,s})]([nt/T]T/n), \ \ t\in [0,T].$$
Define the shortfall risk
\begin{equation*}
\cR_n(T,s,\pi,\sigma)=\sup_{\tau\in\mathcal T^{(n)}_T} \mathbb{E}_{\mathbb P_n}\left(X^{n,s}_{\sigma}\mathbb{I}_{\sigma<\tau}+
Y^{n,s}_{\tau}\mathbb{I}_{\tau\leq\sigma}-V^\pi_{\sigma\wedge\tau}\right)^{+},
\end{equation*}
where $\mathcal{T}^{(n)}_T$ is the set of all stopping times less than $T$.
The shortfall risk for an initial position $(z,y)$ is given by
$R_n(T,s,z,y)=\inf_{(\pi,\sigma)\in\mathcal{A}^{(n)}(T,s,z,y)}
\cR_n(T,s,\pi,\sigma).$

Next, we introduce a simple form of Skorokhod embedding which allows
to consider the above binomial markets
and the BS model on the same probability space.
Set $W^{*}_t=\frac{\ln S^{(s)}_t-\ln s }{\kappa}$, $t\geq 0$,
and for any $n\in\mathbb{N}$ define
recursively
$\theta^{(n)}_0=0$, $\theta^{(n)}_{k+1}=\inf{\{t>\theta^{(n)}_k
:|W^{*}_t-W^{*}_{\theta^{(n)}_k}|=\sqrt\frac{T}{n}\}}$.
Observe (see \cite{DK2}) that for any $k$,
$W^{*}_{\theta^{(n)}_{k+1}}-W^{*}_{\theta^{(n)}_k}$ is independent
of $\mathcal{F}_{\theta^{(n)}_k}$ and takes on the values $\sqrt\frac{T}{n}$ and
$-\sqrt\frac{T}{n}$, with probabilities $p^{(n)}$
and $1-p^{(n)}$, respectively. For any $n$, define the map
$\Pi_n:L^{\infty}(\mathcal{F}^{(n)}_T,\mathbb P_n)\rightarrow{L^{\infty}
(\mathcal{F}_{\theta^{(n)}_n} ,\mathbb P)}$
by $\Pi_n(U)=\tilde U$ so
that if $U=f\bigg(\sqrt{\frac{T}{n}}\xi_1,...,\sqrt{\frac{T}{n}}\xi_n\bigg)$
for a function $f$ on $\{\sqrt{\frac{T}{n}},-\sqrt{\frac{T}{n}}\}^n$
then $\tilde
U=f(W^{*}_{{\theta^{(n)}_1}},W^{*}_{{\theta^{(n)}_2}}-W^{*}_{{\theta^{(n)}_1}},...,
W^{*}_{{\theta^{(n)}_n}}-W^{*}_{{\theta^{(n)}_{n-1}}}).$

The map $\Pi_n$ allows to lift hedges from the binomial models to the
BS model.
For an initial position $(z,y)$ denote by $\mathcal{A}^{W,n}(s,z,y)$ the set of all
\textit{admissible} self financing strategies which are
managed on the set $\{0,\theta^{(n)}_1,...,\theta^{(n)}_n\}$ such that after the time $\theta^{(n)}_n$
the number of stocks in the portfolio is $0$.
Namely,
$\pi=(z,y,\{\tilde\gamma_t\}_{t=0}^\infty)\in\mathcal{A}^{W,n}(s,z,y)$
if $\tilde\gamma_t$ is constant on the interval $(\theta^{(n)}_k,\theta^{(n)}_{k+1}]$, $k<n$
and $\tilde\gamma_t\equiv 0$ on $(\theta^{(n)}_n,\infty)$. The portfolio value is given by (\ref{2.4}).
Define the lifting $\Psi_n:\mathcal{A}^{(n)}(T,s,z,y)\rightarrow \mathcal{A}^{W,n}(s,z,y)\times\mathcal{T}_T$,
$\Psi_n(\pi,\sigma)=(\tilde\pi,\tilde\sigma)$ as follows.
Let $\pi=(z,y,\gamma)$. Then $\tilde\pi=(z,y,\tilde\gamma)$
where
\begin{equation*}
\tilde\gamma_t=y\mathbb{I}_{t=0}+\sum_{i=0}^{n-1} \mathbb{I}_{\theta^{(n)}_i< t\leq\theta^{(n)}_{i+1}}\Pi_n(\gamma_{(i+1)T/n}).
\end{equation*}
Similar arguments as in \cite{D} (see Section 2 there) yield that
\begin{equation*}
V^{\tilde\pi}_{\theta^{(n)}_k}=\Pi_n(V^{\pi}_{kT/n})  \ \ k=0,1,...,n
\end{equation*}
and $V^{\tilde\pi}_t\geq 0$, $t\geq 0$.
Furthermore, the portfolio value
$V^{\tilde\pi}_t$ is constant after $\theta^{(n)}_n$. Observe that if we restrict the portfolio $\tilde\pi$ to the interval $[0,T]$ we
get an element in $\mathcal{A}(T,s,z,y)$.
Next, we
define $\tilde\sigma\in\mathcal{T}_{T}$ by
\begin{equation*}
\tilde\sigma=T\wedge\theta^{(n)}_{\Pi_n(\sigma)}\,\,\mbox{if}\,\,
\Pi_n(\sigma)<n  \ \mbox{and}\,\, \tilde\sigma=T\,\,\mbox{if}\,\,
\Pi_n(\sigma)=n.
\end{equation*}
The following theorem says that the shortfall risk in the BS model can
be approximated by the shortfall risks
in the binomial models defined above.
Furthermore, by lifting the optimal hedges in the binomial models we get "almost" optimal hedges in the BS model.
\begin{thm}\label{thm2.2}
Let $(z,y)\in \mathbb{R}_{++}\times\mathbb{R}$ be an initial position.
There exists a constant $C>0$ such that for any $n\in\mathbb N$
$$|R(T,s,z,y)-R_n(T,s,z,y)|\leq C n^{-1/4} (\ln n)^{3/4}.$$
Furthermore, let $(\pi_n,\sigma_n)\in\mathcal{A}^{(n)}(T,s,z,y)$ be an optimal hedge, i.e.
$\cR_n(T,s,\pi_n,\sigma_n)=R_n(T,s,z,y)$.
Then for the hedges $(\tilde\pi_n,\tilde\sigma_n)=\Psi_n(\pi_n,\sigma_n)$, $n\in\mathbb N$,
we have
$$\cR(T,s,\tilde\pi_n,\tilde\sigma_n)\leq R(T,s,z,y)+ C n^{-1/4} (\ln n)^{3/4},$$
where in the above right hand side we take the restriction of
$\tilde\pi_n$ to the interval [0,T].
\end{thm}
\begin{rem}
Theorems \ref{thm2.1}
--\ref{thm2.2} can be extended to the case where the constant component $\delta$ becomes  a Lipschitz continuous function of time $\delta(t)$. This makes sense since then we can consider constant minimal transaction cost $\delta$ with respect to the original currency where measured by numeraire $\delta(t)=\delta e^{rt}$ with $r$ being the interest rate.
In this case the portfolio value process may no longer be a super--martingale under the martingale measure $\mathbb Q$. The super--martingale property is
used in Lemma \ref{lem3.2}. However, by applying the growth results obtained in Lemma \ref{lem.growth} we can still prove the statement of Lemma \ref{lem3.2}.
In this case the proofs of both Theorems \ref{thm2.1} and \ref{thm2.2} will become more technical and somewhat unwieldy. The first reason for these complications is that
in a setup where the constant component depends on time, the lifting of the hedges more involved and it requires truncation of the portfolio in order to keep the \textit{admissability} condition. The second reason is that
Lemma 4.3 requires a new proof, since the portfolio strategy is not a super--martingale. Both of these steps can be done using Lemma \ref{lem.growth} but in order to simplify exposition we will deal here only with the case where $\delta$ is constant with respect
to the numeraire.
\end{rem}

\section{Proof of Theorem \ref{thm2.1}}\label{sec4}
We start with some preparations.
For any $\mathbf T\leq T$
and $v\in C_{++}[0,T]$ define the continuous stochastic process $\{S^{\mathbf T,v}_t\}_{t=0}^T$ by
$$S^{\mathbf T,v}_t=v_{T-\mathbf T}\mathbb{I}_{t\leq T-\mathbf T}+S^{(v_{T-\mathbf T})}_{t+\mathbf T-T}\mathbb{I}_{t>T-\mathbf T}.$$
Namely, $S^{\mathbf T,v}$ coincides with $v$ on the interval $[0,T-\mathbf T]$ and
$S^{\mathbf T,v}_t$ is a geometric Brownian motion for $t>T-\mathbf T$.
Consider
a cash settled game option with a maturity date $\mathbf T<\infty$ defined in a BS financial market
which is described in Section \ref{sec:2}.
The payoffs are given by
$$Y^{\mathbf T,v}_t=[F(S^{\mathbf T,v})](t+T-\mathbf T) \ \ \mbox{and} \ \
X^{\mathbf T,v}_t=[G(S^{\mathbf T,v})](t+T-\mathbf T) \ \  t\in [0,\mathbf T].$$
Observe that the processes
$X^{\mathbf T,v}_t\geq Y^{\mathbf T,v}_t$, $t\in [0,\mathbf T]$ are continuous and adapted. Furthermore, if $\mathbf T=T$ then
$X^{\mathbf T,v}_t=X^{(v_0)}_t$ and $Y^{\mathbf T,v}_t=Y^{(v_0)}_t$ for $t\in [0,T]$.

Next, we define the shortfall risk for a maturity $\mathbf T\leq T$.
The sets $\mathcal{A}(\mathbf T,s,z,y)$ and $\mathcal T_{\mathbf T}$ are
defined as in Section \ref{sec:2}, just replace $T$ by $\mathbf T$.
For $v\in C_{++}[0,T]$ and a hedge $(\pi,\sigma)\in \mathcal{A}(\mathbf T,v_{T-\mathbf T},z,y)$ the shortfall risk is defined by
$$\cR(\mathbf T,v,\pi,\sigma)=\sup_{\tau\in\mathcal T_{\mathbf T}}\mathbb{E}_{\mathbb P}\left(X^{\mathbf T,v}_{\sigma}\mathbb{I}_{\sigma<\tau}+
Y^{\mathbf T,v}_{\tau}\mathbb{I}_{\tau\leq\sigma}-V^\pi_{\sigma\wedge\tau}\right)^{+}.$$
Similarly to above we set
$R(\mathbf T,v,z,y)=\inf_{(\pi,\sigma)\in\mathcal{A}\left(\mathbf T,v_{T-\mathbf T},z,y\right)}
\cR(\mathbf T,v,\pi,\sigma).$
Observe that $\cR(\mathbf T,v,\cdot,\cdot\cdot)$ and ${R}(\mathbf T,v,\cdot,\cdot\cdot)$ depend only on
$v_{[0,T-\mathbf T]}$. Furthermore for $\mathbf T=T$,
if $v_0=s$ then $\cR(T,v,\cdot,\cdot\cdot)=\cR(T,s,\cdot,\cdot\cdot)$ and
${R}(T,v,\cdot,\cdot\cdot)=R(T,s,\cdot,\cdot\cdot)$.

Now, assume that at a given time a portfolio value is $z$, the number of stocks is
$y$ and the stock price at this moment is $S$. If the investor buys
$\beta\neq 0$ stocks then the new (cash settled) portfolio value will be
\begin{equation}\label{function}
h(S,z,y,\beta):=z+g(y,S)-g(y+\beta,S)-g(\beta,S).
\end{equation}
For $\beta=0$ we define the function $h$ such that it will be continuous in
$\beta=0$. Thus we set
$h(S,z,y,0)=z-\delta$.
Let $\Gamma(S,z,y)$ be the set of all $\beta$ which satisfy $h(S,z,y,\beta)\geq 0$.
 It is clear that $\Gamma(S,z,y)$ is a
 compact set. Observe also that a portfolio strategy is \textit{admissible}
if and only if it consists of permissible trades.
For $y\neq 0$ we have that $-y\in\Gamma(S,z,y)$, and so the set
$\Gamma(S,z,y)$ is not empty. For $y=0$ the set $\Gamma(S,z,y)$ is empty if and only if $z<\delta$.
Define,
\begin{equation}\label{4.new+}
\hat{R}(\mathbf T,v,z,y)=\min\left((X^{\mathbf T,v}_0-z)^{+},
\inf_{\beta\in \Gamma(v_{T-\mathbf T},z,y) }R\left(\mathbf T,v,h(v_{T-\mathbf T},z,y,\beta),\beta+y\right)\right)
\end{equation}
where the infimum over an empty set is $\infty$.

Next, let $s$ be the initial stock price and $(z,y)$ be the initial position of the investor. Set
$$\mathbf V^{s,z,y}_t=z+g(y,s)+
y(S^{(s)}_t-s)-g(y,S^{(s)}_t), \ \ t\geq 0.$$
Observe that
$\mathbf V^{s,z,y}_t$ is the portfolio value at time $t$ of the investor who did
 not trade until this time.

The following result is the first step in the proof of Theorem \ref{thm2.1}.
\begin{lem}\label{lem4.1}
Let $(\mathbf T,v,z,y)\in[0,T]\times C_{++}[0,T]\times\mathbb{R}_{+}\times\mathbb{R}$
and let $s=v_{T-\mathbf T}$ be the initial stock price.
Define the stopping time
$\Theta=\mathbf T\wedge\inf\{t: \mathbf V^{s,z,y}_t<0\}$.
Then,
\begin{eqnarray*}
&R(\mathbf T,v,z,y)\geq
\inf_{\sigma\in\mathcal{T}_{\mathbf T},\sigma\leq\Theta} \sup_{\tau\in\mathcal{T}_{\mathbf T},\tau\leq\Theta}
\mathbb{E}_{\mathbb P}\bigg(\mathbb{I}_{\tau\leq \sigma}\times\\
&(Y^{\mathbf T,v}_{\tau}-\mathbf V^{s,z,y}_{\tau})^{+}+
\mathbb{I}_{\sigma<\tau}\hat{R}(\mathbf T-\sigma,S^{\mathbf T,v},\mathbf V^{s,z,y}_{\sigma},y)\bigg).\nonumber
\end{eqnarray*}
\end{lem}
\begin{proof}
Let $(\pi,\sigma)\in \mathcal{A}(\mathbf T,s,z,y)$.
Set $\sigma_1=\sigma\wedge\min\{t:\gamma_t\neq \gamma_{t+}\}$ where $\pi=(z,y,\gamma)$.
Clearly $\sigma_1\in\mathcal{T}_{\mathbf T}$ is a stopping time.
Introduce the stochastic process
$$U_t=ess \sup_{\tau\geq t, \tau\in\mathcal{T}_{\mathbf T}}\mathbb{E}_{\mathbb P}\left(\left(\mathbb{I}_{\tau\leq\sigma}Y^{\mathbf T,v}_{\tau}
+\mathbb{I}_{\sigma<\tau}X^{\mathbf T,v}_{\sigma}-V^{\pi}_{\sigma\wedge \tau}
\right)^{+}|\mathcal F_t\right), \ t\in [0,\mathbf T].$$
The stochastic process
$\{\mathbb{I}_{t\leq\sigma}Y^{\mathbf T,v}_{t}+\mathbb{I}_{\sigma<t}X^{\mathbf T,v}_{\sigma}-V^{\pi}_{\sigma\wedge t}\}_{t=0}^\mathbf T$
is left continuous with right hand limits. Furthermore, this process is lower semi-continuous from the right.
Thus
from the general theory of optimal stopping (see \cite{KQ} and the references there) it follows that
$\{U_t\}_{t=0}^{\mathbf T}$ is a $c\grave{a}dl\grave{a}g$ process and
for any stopping time
$\rho\leq\mathbf T$,
\begin{equation}\label{4.1++}
U_{\rho}=ess \sup_{\tau\in\mathcal{T}^{\rho}_{\mathbf T}}\mathbb{E}_{\mathbb P}\left(\left(\mathbb{I}_{\tau\leq\sigma}
Y^{\mathbf T,v}_{\tau}+\mathbb{I}_{\sigma<\tau}X^{\mathbf T,v}_{\sigma}-V^{\pi}_{\sigma\wedge \tau}
\right)^{+}|\mathcal F_{\rho}\right)
\end{equation}
where $\mathcal{T}^{\rho}_{\mathbf T}$ is the set of all stopping times $\rho\leq\tau\leq\mathbf T$ which satisfy
$\tau>\rho \mathbb{I}_{\rho<\mathbf T}$.

Clearly, there is no trade
until the time $\sigma_1$. Thus
$V^{\pi}_t=\mathbf V^{s,z,y}_t$ for $t\leq\sigma_1$. This together with the fact that $\sigma_1\leq\sigma$ and (\ref{4.1++}) (for $\rho=\sigma_1$) yields
\begin{equation}\label{4.1+++}
\cR(\mathbf T,v,\pi,\sigma)=\sup_{\tau\in\mathcal{T}_{\mathbf T}}\mathbb{E}_{\mathbb P}\left
(\mathbb{I}_{\sigma_1<\tau}U_{\sigma_1}+\mathbb{I}_{\tau\leq \sigma_1}
(Y^{\mathbf T,v}_{\tau}-\mathbf V^{s,z,y}_{\tau})^{+}\right).
\end{equation}
From the Markov property of the Brownian motion it follows
that
$$U_t\geq R(\mathbf T-t,S^{\mathbf T,v},V^\pi_t,\gamma_t).$$
Thus, by the continuity of $R$ (see Proposition \ref{essential})
it follows that on the event $\sigma_1<\mathbf T$,
\begin{eqnarray}\label{4.1++++}
&U_{\sigma_1}=\mathbb{I}_{\sigma_1=\sigma}\left(X^{\mathbf T,v}_{\sigma_1}-V^{\pi}_{\sigma_1}\right)^{+}
+\mathbb{I}_{\sigma_1<\sigma}\lim_{t\downarrow\sigma_1}U_t\geq\\
&\min\left(\left(X^{\mathbf T,v}_{\sigma_1}-V^{\pi}_{\sigma_1}\right)^{+}, R(\mathbf T-\sigma_1,S^{\mathbf T,v},V^\pi_{\sigma_1+},\gamma_{\sigma_1+})\right)\geq\nonumber\\
&\hat R(\mathbf T-\sigma_1,S^{\mathbf T,v},V^\pi_{\sigma_1},\gamma_{\sigma_1})=\hat R(\mathbf T-\sigma_1,S^{\mathbf T,v},\mathbf V^{s,z,y}_{\sigma_1},y).\nonumber
\end{eqnarray}
From (\ref{4.1+++})--(\ref{4.1++++}) and the inequality $\sigma_1\leq\Theta$ we get
\begin{eqnarray}\label{dynamical}
&\cR(\mathbf T,v,\pi,\sigma)\geq
\inf_{\sigma\in\mathcal{T}_{\mathbf T},\sigma\leq\Theta} \sup_{\tau\in\mathcal{T}_{\mathbf T},\tau\leq\Theta}
\mathbb{E}_{\mathbb P}\bigg(\mathbb{I}_{\tau\leq \sigma}\times\\
&(Y^{\mathbf T,v}_{\tau}-\mathbf V^{s,z,y}_{\tau})^{+}+
\mathbb{I}_{\sigma<\tau}\hat{R}(\mathbf T-\sigma,S^{\mathbf T,v},\mathbf V^{s,z,y}_{\sigma},y)\bigg)\nonumber
\end{eqnarray}
and since $(\pi,\sigma)$ was arbitrary the proof is completed.
\end{proof}
Next, we construct an optimal hedge and verify that it is indeed optimal. The verification will be done
by showing that for the constructed hedge the left hand side
and the right hand side of (\ref{dynamical})
are equal.

Let $(z,y)$ be an initial position.
Set,
$\hat\sigma_0=0$, $\hat\gamma_0=y$, $\hat Z_0=z$.
For $k\geq 1$ define the random time
$\Theta_k=T\wedge\inf\{t\geq\hat\sigma_{k-1}: \mathbf V^{S^{(s)}_{\hat\sigma_{k-1}},\hat Z_{k-1},\hat\gamma_{k-1}}_t<0\}$
and the stochastic process
$\{\mathbf R^{(k)}_t\}_{t=\hat\sigma_{k-1}}^{\Theta_k}$,
\begin{eqnarray*}
&\mathbf R^{(k)}_t=ess\inf_{\sigma\in\mathcal{T}_{T},\hat\sigma_{k-1}\leq\sigma\leq\Theta_k}
ess\sup_{\tau\in\mathcal{T}_{T},\hat\sigma_{k-1}\leq\tau\leq\Theta_k}\\
&\mathbb{E}_{\mathbb P}\left(\mathbb{I}_{\tau\leq\sigma}
\left(Y^{(s)}_{\tau}-\mathbf V^{S^{(s)}_{\hat\sigma_{k-1}},\hat Z_{k-1},\hat\gamma_{k-1}}_{\tau}\right)^{+}+\right.\\
&\left.\mathbb{I}_{\sigma<\tau}\hat{R}\left(T-\sigma,S^{(s)},
\mathbf V^{S^{(s)}_{\hat\sigma_{k-1}},\hat Z_{k-1},\hat\gamma_{k-1}}_{\sigma},\hat\gamma_{k-1}\right)\bigg|
\mathcal{F}_t\right).
 \end{eqnarray*}
Next, introduce the random time
$$\hat\sigma_k=\Theta_k\wedge\inf\left\{t\geq\hat\sigma_{k-1}:\mathbf R^{(k)}_t=
\hat{R}\left(T-t,S^{(s)},\mathbf V^{S^{(s)}_{\hat\sigma_{k-1}},\hat Z_{k-1},\hat\gamma_{k-1}}_t,\hat\gamma_{k-1}\right)\right\}$$
and the random variable
$$\hat\beta_k=\beta^{*}\left(T-\hat\sigma_k,S^{(s)},\mathbf V^{S^{(s)}_{\hat\sigma_{k-1}},\hat Z_{k-1},\hat\gamma_{k-1}}_{\hat\sigma_k},\hat\gamma_{k-1}\right)
$$
where the function $\beta^{*}$ was introduced in Proposition \ref{essential}.
Finally, set
$$\hat\gamma_k=\hat\gamma_{k-1}+\hat\beta_k \ \ \mbox{and} \ \
\hat Z_k=h\left(S^{(s)}_{\hat\sigma_k},\mathbf V^{S^{(s)}_{\hat\sigma_{k-1}},\hat Z_{k-1},\hat\gamma_{k-1}}_{\tau},\hat\gamma_{k-1},\hat\beta_k\right),$$
with
the function $h$
given by (\ref{function}).

The following lemma completes the proof of Theorem \ref{thm2.1} and gives a characterisation of the optimal hedge.
\begin{lem}\label{lem4.2}
For any $k\geq 1$ the stochastic process
$\{\mathbf R^{(k)}_t\}_{t=\hat\sigma_{k-1}}^{\Theta_k}$
is well defined, continuous and $\hat\sigma_k,\Theta_k$ are stopping times. Furthermore, the random variables
$\hat\gamma_k,\hat Z_k$ are $\mathcal{F}_{\hat\sigma_k}$ measurable.
Next, the structure of an optimal hedge is described in the following way.
Set,
\begin{eqnarray*}
&\hat N=\min\{k:\hat\sigma_k=T\}\wedge\\
&\min\left\{k:\hat
R\left(T-\hat\sigma_k,S^{(s)},\mathbf V^{S^{(s)}_{\hat\sigma_{k-1}},\hat Z_{k-1},\hat\gamma_{k-1}}_{\hat\sigma_k},\hat\gamma_{k-1}\right)=\left(X^{(s)}_{\hat\sigma_k}-
\mathbf V^{S^{(s)}_{\hat\sigma_{k-1}},\hat Z_{k-1},\hat\gamma_{k-1}}_{\hat\sigma_k}
\right)^{+}\right\}.
\end{eqnarray*}
Then $\hat N<\infty$ a.s and the hedge $(\hat\pi,\hat\sigma)\in\mathcal{A}\left(T,s,z,y\right)$ which is given by
$\hat\pi=(z,y,\hat\gamma)$, where
$$\hat\gamma_t=y\mathbb{I}_{t=0}+\sum_{i=1}^{\hat N-1} \hat\gamma_{i}\mathbb{I}_{(\hat\sigma_{i},\hat\sigma_{i+1}]} \ \ \mbox{and} \ \ \hat\sigma=\hat\sigma_{\hat N},$$
satisfies
$R(T,s,z,y)=\cR(T,s,\hat\pi,\hat\sigma).$
\end{lem}
\begin{proof}
First, we establish a stronger version of the
inequality
(\ref{dynamical}). The exact form is the following. For any stopping
time $\theta\in\mathcal{T}_{T}$ and random variables $Z\geq 0$, $Y$ which are $\mathcal{F}_{\theta}$
measurable,
\begin{eqnarray}\label{4.100}
&R(T-\theta,S^{(s)},Z,Y)\geq
\inf_{\sigma\in\mathcal{T}_{T},\theta\leq\sigma\leq\Theta} \sup_{\tau\in\mathcal{T}_{T},\theta\leq\tau\leq\Theta}\\
&\mathbb{E}_{\mathbb P}\left(\mathbb{I}_{\tau\leq \sigma}\left(Y^{(s)}_{\tau}-\mathbf V^{S^{(s)}_{\theta},Z,Y}_{\tau}\right)^{+}+
\mathbb{I}_{\sigma<\tau}
\hat{R}\left(T-\sigma,S^{(s)},\mathbf V^{S^{(s)}_{\theta},Z,Y}_{\sigma},Y\right)\right),\nonumber
\end{eqnarray}
where
$\Theta=T\wedge\inf\{t\geq\theta:\mathbf V^{S^{(s)}_{\theta},Z,Y}_{t}< 0\}.$
In order to derive (\ref{4.100}) introduce the stochastic processes
\begin{eqnarray*}
&\mathbf X_t:=\hat{R}\left(T-t,S^{(s)},\mathbf V^{S^{(s)}_{\theta},Z,Y}_{t},Y\right), \ \ t\in[\theta,\Theta]\\
&\mathbf Y_t:=\left(Y^{(s)}_{t}-\mathbf V^{S^{(s)}_{\theta},Z,Y}_{t}\right)^{+}, \ \ t\in[\theta,\Theta].
\end{eqnarray*}
From Proposition \ref{essential}
it follows that the above processes are continuous.
Observe that (since the buyer can stop at $0$)
for any $(\hat{\mathbf T},\hat v,\hat z,\hat y)\in [0,T]\times C_{++}[0,T]\times\mathbb R_{+}\times\mathbb R$,
$\hat R(\hat {\mathbf T},\hat v,\hat z,\hat y)\geq (Y^{\hat {\mathbf T},\hat v}_0-\hat z)^{+}.$
We conclude that $\mathbf X\geq \mathbf Y$,
and so by applying standard results on Dynkin games (see \cite{LM})
we derive (\ref{4.100}) from (\ref{dynamical}).

Next for
a given $k$, consider the stochastic processes
\begin{eqnarray*}
&\hat{\mathbf X}_t:=\hat{R}\left(T-t,S^{(s)},\mathbf V^{S^{(s)}_{\hat\sigma_{k-1}},\hat Z_{k-1},\hat\gamma_{k-1}}_{t},
,\hat\gamma_{k-1}\right), \ \ t\in [\hat\sigma_{k-1},\Theta_k],\\
&\hat{\mathbf Y}_t=\left(Y^{(s)}_{t}-\mathbf V^{S^{(s)}_{\hat\sigma_{k-1}},\hat Z_{k-1},\hat\gamma_{k-1}}_{t}\right)^{+}, \ \ t\in [\hat\sigma_{k-1},\Theta_k]
\end{eqnarray*}
where $\hat Z_k,\hat\gamma_k,\hat\sigma_k,\Theta_k$, $k\in\mathbb N$ were define above.
Observe that $\hat{\mathbf X}\geq \hat {\mathbf Y}$. Thus,
(induction on $k$) by applying Proposition 3.9 from \cite{IK} (see also Theorem 4.1 in \cite{CK})
we conclude that $\{\mathbf R^{(k)}_t\}_{t=\hat\sigma_{k-1}}^{\Theta_k}$ is a continuous stochastic process
and $\hat\sigma_k$ is a stopping time, and
$\hat\gamma_k,\hat Z_k$ are $\mathcal{F}_{\hat\sigma_k}$ measurable.
From the definition of the function $\hat R$ it follows that $\hat Z_k\geq 0$ for $k\leq\hat N$.
This together with the definition of the stopping times $\Theta_k$, $k\in\mathbb N$
and the fact that the portfolio value is constant after $\hat\sigma_{\hat N}$ yields that
$\hat\pi$ is an \textit{admissible} portfolio.
Observe, that for any $n\in\mathbb N$ we have
$$0\leq V^{\hat\pi}_{\hat\sigma_n}\leq z+\max(\delta,\mu |y| s)+\sum_{i=1}^n
\gamma^{\pi}_{\hat\sigma_{i-1}}(S^{(s)}_{\hat\sigma_i}-S^{(s)}_{\hat\sigma_{i-1}})-n\delta\mathbb{I}_{\hat N>n}.$$
Taking the expectation with respect to the martingale measure $\mathbb Q$
we obtain that
$\mathbb Q(\hat N>n)\leq\frac{ z+\max(\delta,\mu |y| s)}{n},$
and so $\hat N<\infty$ a.s.
Thus $(\hat\pi,\hat\sigma)\in\mathcal{A}(T,s,z,y)$.

Finally, we prove that $(\hat\pi,\hat\sigma)$ is an optimal hedge.
Choose $\epsilon>0$. There exists $N\in \mathbb N$ such that
\begin{equation}\label{4.5}
\mathbb P(\hat N>N)<\epsilon.
\end{equation}
Define the continuous martingale
$M_t=\mathbb E_{\mathbb P}\left(\sup_{0\leq t\leq  T}\left(X^{(s)}_t\right)^2\mathbb {I}_{\hat N>N}|\mathcal F_t\right)$, $t\in [0,T]$,
and the stochastic process
$$\hat U_t=ess \sup_{\tau\geq t, \tau\in\mathcal{T}_{T}}\mathbb{E}_{\mathbb P}\left(\left(\mathbb{I}_{\tau\leq\hat\sigma}Y^{(s)}_{\tau}
+\mathbb{I}_{\hat\sigma<\tau}X^{(s)}_{\hat\sigma}-V^{\hat\pi}_{\hat\sigma\wedge \tau}
\right)^{+}|\mathcal F_t\right), \ t\in [0,T].$$

Let us show that for any $0\leq k\leq N$,
\begin{equation}\label{4.6}
\hat U_{\hat\sigma_{k\wedge\hat N}}\leq \hat{R}\left(T-\hat\sigma_{k\wedge\hat N},S^{(s)},V^{\hat\pi}_{\hat\sigma_{k\wedge\hat N}},\hat\gamma_{k\wedge\hat N-1}\right)+
\sqrt{M_{\hat\sigma_{k\wedge\hat N}}},
\end{equation}
where for $k=0$ we set
$\hat{R}(T,S^{(s)}, z,\hat\gamma_{-1}):=R(T,s,z,y)$.

We start with $k=N$. From the definition
of $\hat N$ and the Jensen inequality
it follows that
\begin{eqnarray*}
&\hat U_{\hat\sigma_{N\wedge\hat N}}\leq \mathbb{I}_{\hat\sigma_{N\wedge\hat N}=\hat\sigma}
 \left(X^{(s)}_{\hat\sigma_{N\wedge\hat N}}-V^{\hat\pi}_{\hat\sigma_{N\wedge\hat N}}\right)^{+}+
\mathbb{I}_{N<\hat N} \mathbb E_{\mathbb P}\left(\sup_{0\leq t\leq \mathbb T}\left(X^{(s)}_t\right)|\mathcal F_{\hat\sigma_{N\wedge\hat N}}\right)
\\
&\leq\hat{R}\left(T-\hat\sigma_{N\wedge\hat N},S^{(s)},V^{\hat\pi}_{\hat\sigma_{N\wedge\hat N}},\hat\gamma_{N\wedge\hat N-1}\right)+
\sqrt{M_{\hat\sigma_{N\wedge\hat N}}}.
\end{eqnarray*}
Next, we prove that if (\ref{4.6}) is true for $k+1$ then it is true for $k$.
From the definition of $\hat N$, on the event $k\geq\hat N$ (which is $\mathcal{F}_{\hat\sigma_{k\wedge\hat N}}$ measurable),
(\ref{4.6}) trivially holds true. Consider the event
$k<\hat N$. On this event, similarly
to (\ref{4.1+++}) we have
\begin{eqnarray*}
&\hat U_{\hat\sigma_{k}}=\sup_{\tau\in\mathcal{T}_{T}}\mathbb{E}_{\mathbb P}\left(\mathbb{I}_{\hat\sigma_{k+1}<\tau}
\hat U_{\hat\sigma_{k+1}}+
\mathbb{I}_{\tau\leq \hat\sigma_{k+1}}\left(Y^{(s)}_{\tau}-\mathbf{V}^{S^{(s)}_{\hat\sigma_k},\hat Z_k,\hat\gamma_k}_{\tau}
\right)^{+}|\mathcal{F}_{\hat\sigma_{k\wedge \hat N}}\right)\leq
\end{eqnarray*}
from the induction assumption, since $\sqrt M$ is a super--martingale and by
the definition of $\hat\sigma_{k+1}$,
\begin{eqnarray*}
&\leq\sup_{\tau\in\mathcal{T}_{T}}\mathbb{E}_{\mathbb P}\left(\mathbb{I}_{\tau\leq \hat\sigma_{k+1}}\left(Y^{(s)}_{\tau}-
\mathbf{V}^{S^{(s)}_{\hat\sigma_k},\hat Z_k,\hat\gamma_k}_{\tau}\right)^{+}+
\mathbb{I}_{\hat\sigma_{k+1}<\tau}\times\right.\\
&\left.\hat{R}\left(T-\hat\sigma_{k+1},S^{(s)},V^{\hat\pi}_{\hat\sigma_{k+1}},\hat\gamma_{k}\right)
\bigg|\mathcal{F}_{\hat\sigma_{k\wedge\hat N}}\right)+\sqrt{M_{\hat\sigma_{k\wedge\hat N}}}=\\
&\sqrt{M_{\hat\sigma_{k\wedge\hat N}}}+ess\inf_{\sigma\in\mathcal{T}_{ T},\hat\sigma_{k}\leq\sigma\leq\Theta_{k+1}}
ess\sup_{\tau\in\mathcal{T}_{ T},\hat\sigma_{k}\leq\tau\leq\Theta_{k+1}}\\
&\mathbb{E}_{\mathbb P}\left(\mathbb{I}_{\tau\leq\sigma}\left(Y^{(s)}_{\tau}-
\mathbf{V}^{S^{(s)}_{\hat\sigma_k},\hat Z_k,\hat\gamma_k}_{\tau}\right)^{+}+
\mathbb{I}_{\sigma<\tau}\hat{R}\left(T-\sigma,S^{(s)},\mathbf{V}^{S^{(s)}_{\hat\sigma_k},\hat Z_k,\hat\gamma_k}_{\sigma},\hat\gamma_k\right)\bigg|
\mathcal{F}_{\hat{\sigma}_{k\wedge\hat N}}\right)\leq
\end{eqnarray*}
by (\ref{4.100}),
$$\leq \sqrt{M_{\hat\sigma_{k\wedge\hat N}}}+R(T-\hat\sigma_k,S^{(s)},V^{\hat\pi}_{\hat\sigma_k},\hat\gamma_k)
= \sqrt{M_{\hat\sigma_{k\wedge\hat N}}}+\hat{R}\left(T-\hat\sigma_{k\wedge\hat N},S^{(s)},V^{\hat\pi}_{\hat\sigma_k},\hat\gamma_{k\wedge\hat N-1}\right)
$$
where the last equality
follows from the definition of $\hat\gamma_k$ and the fact that we are on the event $k<\hat N$.
This completes the proof of (\ref{4.6}).

From (\ref{2.3}), the Cauchy--Schwarz inequality and (\ref{4.5}) it follows that $M_0=O(\epsilon^{1/2}).$
Thus
by applying (\ref{4.6}) for $k=0$,
$$\cR(T,s,\hat\pi,\hat\sigma)=\hat U_0\leq \sqrt{M_0}+R(T,s,z,y)=O(\epsilon^{1/4})+R(T,s,z,y),$$
and by taking $\epsilon\downarrow 0$ we complete the proof.
\end{proof}

\section{Proof of Theorem \ref{thm2.2}}
\label{sec2+}
First, we follow Lemma 4.1 in \cite{D} and obtain a bound on the growth of \textit{admissible} portfolios.
\begin{lem}\label{lem.growth}
Let $(z,y)\in \mathbb{R}_{++}\times\mathbb{R}$ be an initial position. There exists a constant $\tilde C$ such that for any
$\pi\in\mathcal{A}(T,s,z,y)$,
$$\mathbb E_{\mathbb Q} \left(\max_{0\leq t\leq T}|\gamma_t| S^{(s)}_t+\int_{0}^T S^{(s)}_u |d\gamma_u|\right)^2\leq \tilde C(1+z^2+y^2).$$
\end{lem}
\begin{proof}
Let $\pi=(z,y,\gamma)\in \mathcal{A}(T,s,z,y)$.
We will use the integration by parts formula
\begin{equation}\label{parts}
\int_{0}^T \gamma_u dS^{(s)}_u=\gamma_tS^{(s)}_t-ys-\int_{[0,t]}S^{(s)}_u d\gamma_u,
\end{equation}
and the decomposition
$\gamma_t=\gamma^{+}_t-\gamma^{-}_t$
into a positive variation $\gamma^{+}$
and a negative variation $\gamma^{-}$.
From (\ref{2.4}) and (\ref{parts}) it follows that for any $t\in [0,T]$
$$0\leq V^\pi_t\leq z+g(y,s)+|y|s-(1+\mu)\int_{[0,t]}S^{(s)}_ud\gamma_u-2\mu\int_{[0,t]} S^{(s)}_ud\gamma^{-}_u+(1+\mu)\gamma_t S^{(s)}_t$$
and
$$0\leq V^\pi_t\leq z+g(y,s)+|y|s-(1-\mu)\int_{[0,t]}S^{(s)}_ud\gamma_u-2\mu \int_{[0,t]} S^{(s)}_ud\gamma^{+}_u+(1-\mu)\gamma_t S^{(s)}_t.$$
This together with (\ref{parts}) yields
\begin{eqnarray}\label{cal0}
&\int_{0}^t S^{(s)}_ud\gamma^{-}_u\leq \frac{1}{2\mu}\left(z+g(y,s)+|y|s+(1+\mu)(|y| s+\int_{0}^t \gamma_udS^{(s)}_u)\right) \\
&\mbox{and} \ \ \int_{0}^t S^{(s)}_ud\gamma^{+}u\leq \frac{1}{2\mu}\left(z+g(y,s)+|y|s+(1-\mu)(|y|s+\int_{0}^t \gamma_udS^{(s)}_u)\right).\nonumber
\end{eqnarray}
Consequently,
\begin{eqnarray}\label{cal1}
&\gamma_tS^{(s)}_t
\geq\int_{0}^t\gamma_u dS^{(s)}_u-\int_{[0,t]} S^{(s)}_u d\gamma^{-}_u-|y|s\geq \\
&-\frac{1}{2\mu}(z+g(y,s)+(2+3 \mu)|y| s)-\frac{1-\mu}{2\mu} \int_{0}^t\gamma_uS^{(s)}_u  \nonumber
\end{eqnarray}
and
\begin{eqnarray}\label{cal2}
&\gamma_tS^{(s)}_t\leq
|y|s+\int_{0}^t\gamma_u dS^{(s)}_u+\int_{[0,t]} S^{(s)}_ud\gamma^{+}_u\leq \\
&\frac{1}{2\mu}(z+g(y,s)+(2+ \mu)|y| s)-\frac{1+\mu}{2\mu} \int_{0}^t\gamma_uS^{(s)}_u.\nonumber
\end{eqnarray}
From (\ref{cal1})--(\ref{cal2}) and the inequality $(a+b)^2\leq 2(a^2+b^2)$ we obtain
\begin{equation}\label{cal3}
|\gamma_tS^{(s)}_t|^2\leq \frac{\left(z+g(y,s)+(2+3 \mu)|y| s\right)^2}{2\mu^2}+\frac{(1+\mu)^2}{2\mu^2}
\left(\int_{0}^t\gamma(u)dS(u)\right)^2.
\end{equation}
Following the arguments of Lemma 4.1 in \cite{D} we see that there is a constant $\tilde c$ such that
$$\mathbb E_{\mathbb Q} \left(\sup_{0\leq t\leq T}\left(\int_{0}^t\gamma_udS^{(s)}_u\right)^2\right)\leq \tilde c (z+g(y,s)+(2+3 \mu)|y| s)^2.$$
This together with (\ref{cal0}) and (\ref{cal3}) competes the proof.
\end{proof}

Next, for any $n\in\mathbb N$ introduce the piecewise constant stochastic process
$\hat S^{n,s}_t=S^{(s)}_{\theta^{(n)}_{[nt/T]}}$, $t\in [0,T]$
and the payoffs
$$\hat Y^{n,s}_t=[F(\hat S^{n,s})]([nt/T]T/n)\leq \hat X^{n,s}_t=[G(\hat S^{n,s})]([nt/T]T/n) \ \  t\in [0,T].$$
By applying the results of Section 4 in \cite{K} we get that there exists a constant $C_1$ such that for any $n$,
\begin{equation}\label{2+.101}
\mathbb E_{\mathbb P}\left(\max_{1\leq k\leq n}\max_{\theta^{(n)}_{k-1}\leq t\leq\theta^{(n)}_k}
\left(|X^{(s)}_{t\wedge T}-\hat X^{n,s}_{\frac{kT}{n}}|+|Y^{(s)}_{t\wedge T}-\hat Y^{n,s}_{\frac{kT}{n}}|\right)\right)\leq C_1 n^{-1/4} (\ln n)^{3/4}.
\end{equation}
Let $\hat{\mathcal T}_n$ the set of all stopping times
with respect to the filtration ${\{\mathcal F_{\theta^{(n)}_k}\}}_{k=0}^n$ with values
$k=0,1,...,n$.
Set,
\begin{equation*}
\hat R_n(T,s,z,y)=\inf_{\pi\in\mathcal{A}^{W,n}(s,z,y)}\inf_{\zeta\in \hat{\mathcal T}_n}
\sup_{\eta\in \hat{\mathcal T}_n}\mathbb{E}_{\mathbb P}\left(\hat X^{n,s}_{\frac{\zeta T}{n}}\mathbb{I}_{\zeta<\eta}+
\hat Y^{n,s}_{\frac{\eta T}{n}}\mathbb{I}_{\eta\leq\zeta}-V^\pi_{\theta^{(n)}_{\zeta\wedge\eta}}\right)^{+}.
\end{equation*}
Observe that $\hat S^{n,s}_t=\Pi_n (S^{n,s}_t)$,
$\hat Y^{n,s}_t=\Pi_n (Y^{n,s}_t)$ and
$\hat X^{n,s}_t=\Pi_n (X^{n,s}_t)$, $t\in [0,T]$. Thus by using similar arguments as in \cite{D} and \cite{DK2} (see Section 3 in these papers) we get
that
\begin{equation}\label{2+.new}
\hat R_n(T,s,z,y)=R_n(T,s,z,y).
\end{equation}
The equality (\ref{2+.new}) plays a key role in the proof of the following result.
\begin{lem}\label{lem2+.1}
Let $(z,y)\in \mathbb{R}_{++}\times\mathbb{R}$ be an initial position. There exists a constant $C_2>0$ such that for any
$n\in\mathbb N$
\begin{equation*}
R(T,s,z,y)\geq R_n(T,s,z,y)-C_2 n^{-1/4} (\ln n)^{3/4}.
\end{equation*}
\end{lem}
\begin{proof}
Fix $n$ and set $z_n=z-n^{-1/3}$. Assume that $n$ is sufficiently large, so $z_n\geq 0$.
Let $(\pi,\sigma)\in\mathcal{A}(T,s,z_n,y)$ be such that
 \begin{equation}\label{2+.newww}
R(T,s,z_n,y)>\cR(T,s,\pi,\sigma)-\frac{1}{n}.
\end{equation}
Introduce the random variable
$$\Gamma=\max_{0\leq t\leq T}|\gamma_t| S^{(s)}_t+\int_{0}^T S^{(s)}_u |d\gamma_u|.$$
Define the stopping time
$$\Upsilon=T\wedge\inf\left\{t: \max_{0\leq u\leq t}|\gamma_u| S^{(s)}_u+\int_{0}^t S^{(s)}_u |d\gamma_u|>n^{1/7}\right\},$$
and the portfolio $\dot\pi=(z,y,\dot\gamma)$ by
$\dot\gamma_t=\gamma_t\mathbb{I}_{t\leq\Upsilon}$. Namely, we liquidate the portfolio at the stopping time $\Upsilon$. Observe that the
initial capitals of the portfolios $\dot\pi$ and $\pi$, equal
to $z$ and $z_n$, respectively.
From Lemma \ref{lem.growth} and the Chebyshev inequality it follows that
$$\mathbb Q(\Upsilon<T)=\mathbb Q(\Gamma> n^{1/7})=O(n^{-2/7}).$$
This together with (\ref{2.martingale})--(\ref{2.3}) and the H\"older inequality (for $p=8$, $q=8/7$) gives
\begin{equation}\label{2+.0}
\cR(T,s,\dot{\pi},\sigma)-\cR(T,s,\pi,\sigma)\leq
\mathbb E_{\mathbb Q}\left(Z^{-1}_T\sup_{0\leq t\leq T}X^{(s)}_t\mathbb{I}_{\Upsilon<T}\right)=O(n^{-1/4}).
\end{equation}

Introduce the
portfolio $\tilde\pi=(z,y,\tilde\gamma)\in\mathcal{A}^{W,n}(s,z,y)$ which is managed at the
stopping times $0,\theta^{(n)}_1,...,\theta^{(n)}_n$ and is given by
$$\tilde\gamma_t=y\mathbb{I}_{t=0}+\sum_{i=0}^{n-1}\dot\gamma_{\theta^{(n)}_i}\mathbb{I}_{(\theta^{(n)}_i,\theta^{(n)}_{i+1}]}.$$
Let $t\in [0,T]$ and consider the event
$t\in (\theta^{(n)}_k,\theta^{(n)}_{k+1}]$ for some $k<n$.
From the integration by parts
formula we get
\begin{eqnarray*}
&V^{\dot\pi}_{t}=z_n
+\dot\gamma_t S^{(s)}_t-ys-
\left(\sum_{i=1}^k \int_{[\theta^{(n)}_{i-1},\theta^{(n)}_i)}S^{(s)}_u d\dot\gamma_u-
\int_{[\theta^{(n)}_{k},t)}S^{(s)}_u d\dot\gamma_u\right)-\\
&\left(\sum_{i=1}^k\sum_{u\in[\theta^{(n)}_{i-1},\theta^{(n)}_i)}g(\dot\gamma_{u+}-\dot\gamma_u, S^{(s)}_u)
+\sum_{u\in[\theta^{(n)}_{k},t)}g(\dot\gamma_{u+}-\dot\gamma_u, S^{(s)}_u)\right)\leq
\end{eqnarray*}
since $|S^{(s)}_u-S^{(s)}_{\theta^{(n)}_{k+1}}|
\leq 4\kappa\sqrt\frac{T}{n}\min(S^{(s)}_u,S^{(s)}_{\theta^{(n)}_{k+1}})$ for $u\in [\theta^{(n)}_k,\theta^{(n)}_{k+1}]$ and large $n$
\begin{eqnarray*}
&\leq z_n+
\dot\gamma_{t} S^{(s)}_{\theta^{(n)}_{k+1}}+4\kappa\dot\gamma_t\sqrt\frac{T}{n}S^{(s)}_{t}-ys-\\
&\sum_{i=1}^{k+1} S^{(s)}_{\theta^{(n)}_{i}}\int_{[\theta^{(n)}_{i-1},\theta^{(n)}_i)} d\dot\gamma_u
+4\kappa \sqrt\frac{T}{n}\int_{0}^T S^{(s)}_u |d\dot\gamma_u|-\\
&\sum_{i=1}^k g(\dot{\gamma}_{\theta^{(n)}_{i}}-\dot{\gamma}_{\theta^{(n)}_{i-1}}, S^{(s)}_{\theta^{(n)}_i})
+4\kappa\mu \sqrt\frac{T}{n}\int_{0}^T S^{(s)}_u |d\dot\gamma_u|-
\int_{[\theta^{(n)}_k,t)}\mu S^{(s)}_t|d\dot\gamma_t|\leq
\end{eqnarray*}
since
$\max_{0\leq u\leq T}|\dot\gamma_u| S^{(s)}_u+\int_{0}^T S^{(s)}_u |d\dot\gamma_u|\leq n^{1/7}$, then for large $n$
\begin{eqnarray*}
&\leq z-ys +\dot\gamma_{\theta^{(n)}_{k}} S^{(s)}_{\theta^{(n)}_{k+1}}+
S^{(s)}_{\theta^{(n)}_{k+1}}\int_{[\theta^{(n)}_{k},t)}|d\dot\gamma_u|-
\sum_{i=1}^{k+1} S^{(s)}_{\theta^{(n)}_{i}}\int_{[\theta^{(n)}_{i-1},\theta^{(n)}_i)} d\dot\gamma_u\\
&-\sum_{i=1}^k g(\dot{\gamma}_{\theta^{(n)}_{i}}-\dot{\gamma}_{\theta^{(n)}_{i-1}}, S^{(s)}_{\theta^{(n)}_i})-
\int_{[\theta^{(n)}_k,t)}\mu S^{(s)}_t|d\dot\gamma_t|\leq\\
&z-ys +\dot\gamma_{\theta^{(n)}_{k}} S^{(s)}_{\theta^{(n)}_{k+1}}-
\sum_{i=1}^{k+1} S^{(s)}_{\theta^{(n)}_{i+1}}\int_{[\theta^{(n)}_{i-1},\theta^{(n)}_i)} d\dot\gamma_u\\
&-
\sum_{i=1}^k g(\dot{\gamma}_{\theta^{(n)}_{i}}-\dot{\gamma}_{\theta^{(n)}_{i-1}}, S^{(s)}_{\theta^{(n)}_i})=
V^{\tilde\pi}_{\theta^{(n)}_{k+1}},
\end{eqnarray*}
where the last equality follows from the summation by parts.
We conclude that (for sufficiently large $n$)
\begin{equation}\label{2+.10}
V^{\tilde\pi}_{\theta^{(n)}_{k+1}}\geq V^{\dot\pi}_t, \ \ t\in (\theta^{(n)}_{k},\theta^{(n)}_{k+1}], \ \ k<n.
\end{equation}
In particular, $\tilde\pi\in\mathcal{A}^{W,n}(s,z,y)$ is an \textit{admissible} portfolio.

Let $\zeta\in\hat{\mathcal T_n}$ be given by
\begin{equation}\label{2+.9}
\zeta=n\wedge\min\{k:\theta^{(n)}_k\geq\sigma\}\,\,\mbox{if}\,\, \sigma<T
\ \mbox{and}\,\, \zeta=n\,\,\mbox{if}\,\,\sigma=T
\end{equation}
where $(\pi,\sigma)$ satisfies (\ref{2+.newww}).
From (\ref{2+.new})--(\ref{2+.0}) and Lemma \ref{lem3.2} we get
\begin{eqnarray*}
&R_n(T,s,z,y)-R(T,s,z,y)\leq O(|z_n-z|^{3/4})+\frac{1}{n}+O(n^{-1/4})+\\
&\hat R_n(T,s,z,y)-R(T,s,\dot\pi,\sigma)\leq O(n^{-1/4})+\\
&\sup_{\eta\in \hat{\mathcal T}_n}\mathbb{E}_{\mathbb P}\left(\hat X^{n,s}_{\frac{\zeta T}{n}}\mathbb{I}_{\zeta<\eta}+
\hat Y^{n,s}_{\frac{\eta T}{n}}\mathbb{I}_{\eta\leq\zeta}-V^{\tilde\pi}_{\theta^{(n)}_{\zeta\wedge\eta}}\right)^{+}-
\\
&\sup_{\tau\in\mathcal{T}_T} \mathbb{E}_{\mathbb P}\left(X^{(s)}_{\sigma}\mathbb{I}_{\sigma<\tau}+
Y^{(s)}_{\tau}\mathbb{I}_{\tau\leq\sigma}-V^{\dot\pi}_{\sigma\wedge\tau}\right)^{+}\leq
\end{eqnarray*}
since $T\wedge\theta^{(n)}_{\eta}\in\mathcal{T}_T$ for any $\eta\in\mathcal{T}_n$
\begin{eqnarray*}
&\leq O(n^{-1/4})+ \sup_{\eta\in \hat{\mathcal T}_n}\mathbb{E}_{\mathbb P}\left(\hat X^{n,s}_{\frac{\zeta T}{n}}\mathbb{I}_{\zeta<\eta}+
\hat Y^{n,s}_{\frac{\eta T}{n}}\mathbb{I}_{\eta\leq\zeta}-V^{\tilde\pi}_{\theta^{(n)}_{\zeta\wedge\eta}}\right)^{+}-\\
&\sup_{\eta\in \hat{\mathcal T}_n} \mathbb{E}_{\mathbb P}\left(X^{(s)}_{\sigma}\mathbb{I}_{\sigma<\theta^{(n)}_{\eta}\wedge T}+
Y^{(s)}_{\theta^{(n)}_{\eta}\wedge T}\mathbb{I}_{\theta^{(n)}_{\eta}\wedge T\leq\sigma}-V^{\dot\pi}_{\sigma\wedge\theta^{(n)}_{\eta}}\right)^{+}\leq
\end{eqnarray*}
since $\sigma<\theta^{(n)}_{\eta}\wedge T$ by (\ref{2+.9}) if $\zeta<\eta$
\begin{eqnarray*}
&\leq O(n^{-1/4})+\mathbb{E}_{\mathbb P}\left(|\hat X^{n,s}_{\frac{\zeta T}{n}}-X^{(s)}_{\sigma}|\mathbb{I}_{\sigma<\theta^{(n)}_n}\right)+
\sup_{\eta\in \hat{\mathcal T}_n}\mathbb{E}_{\mathbb P}|\hat Y^{n,s}_{\frac{\eta_n T}{n}}-Y^{(s)}_{\theta^{(n)}_{\eta}\wedge T}|+\\
&\sup_{\eta\in \hat{\mathcal T}_n}\mathbb{E}_{\mathbb P}
(V^{\dot\pi}_{\sigma\wedge\theta^{(n)}_{\eta}}-V^{\tilde\pi}_{\theta^{(n)}_{\zeta\wedge\eta}})^{+}\leq
\end{eqnarray*}
since $V^{\dot\pi_n}_{\theta^{(n)}_{\zeta\wedge\eta}}\geq V^{\dot\pi}_{\sigma\wedge\theta^{(n)}_{\eta}}$
by (\ref{2+.10}), together with (\ref{2+.101}),
\begin{eqnarray*}
&\leq O(n^{-1/4})+ \mathbb E_{\mathbb P}\left(\max_{1\leq k\leq n}\max_{\theta^{(n)}_{k-1}\leq t\leq\theta^{(n)}_k}
\left(|X^{(s)}_{t\wedge T}-\hat X^{n,s}_{\frac{kT}{n}}|+|Y^{(s)}_{t\wedge T}-\hat Y^{n,s}_{\frac{kT}{n}}|\right)\right)\\
&=O(n^{-1/4} (\ln n)^{3/4}),
\end{eqnarray*}
and the result follows.
\end{proof}
In view of Lemma \ref{lem2+.1}, in order to complete the proof of Theorem \ref{thm2.2} it remains to establish
the following result.
\begin{lem}
There exists a constant $C_3$ such that for any $n\in\mathbb N$ and $(\pi,\sigma)\in\mathcal{A}^{(n)}(T,s,z,y)$,
$$\cR(T,s,\tilde\pi,\tilde\sigma)\leq \cR_n(T,s,\pi,\sigma)+ C_3 n^{-1/4} (\ln n)^{3/4},$$
for $(\tilde\pi,\tilde\sigma)=\Psi_n(\pi,\sigma)$.
\end{lem}
\begin{proof}
The proof follows the proof of (2.26) in \cite{DK2}.
In the proof of (2.26) in \cite{DK2} we showed that if we lift a hedge to the BS model, the shortfall risk can increase only by the amount
$O(n^{-1/4} (\ln n)^{3/4})$.
Though in \cite{DK2}
there is no friction, the only property that we used there is that
the portfolio value process is a super--martingale with respect to the martingale
measure $\mathbb Q$. In the current setup, this fact remains true,
and so we just follow the proof from \cite{DK2}.
\end{proof}

\section{Regularity properties of shortfall risk}\label{sec3}
In this section we do not assume Lipschitz continuity of the functions $F,G$ (that is (\ref{condition})) but
just continuity and (\ref{2.3}).
We start with the following lemma.
\begin{lem}\label{lem3.time}
Let $v\in C_{++}[0,T]$.
There exists a continuous function $m_v:\mathbb {R}_{+}\rightarrow\mathbb{R}_{+}$ (modulus of continuity)
with $m_v(0)=0$ such that for any $\mathbf T_1,\mathbf T_2\in [0, T]$,
$z\geq 0$
and $y\in \mathbb R$, we have
$$\left|{R}(\mathbf T_1,v,z,y)-R\left(\mathbf T_2,v,z,\frac{y v(T-\mathbf T_1)}{v(T-\mathbf T_2)}\right)\right|\leq m_v(|\mathbf T_1-\mathbf T_2|).$$
\end{lem}
\begin{proof}
Let $\mathbf T_1,\mathbf T_2\in [0, T]$,
$z\geq 0$ and $y\in\mathbb R$.
Without loss of generality, we assume that $\mathbf T_1<\mathbf T_2$. Choose $\epsilon>0$.
There exists a hedge $(\pi_1,\sigma_1)\in \mathcal{A}\left(\mathbf T_1,v(T-\mathbf T_1),z,y\right)$ such that
\begin{equation}\label{3.200}
R(\mathbf T_1,v,z,y)> \cR(\mathbf T_1,v,\pi_1,\sigma_1)-\epsilon.
\end{equation}
Set $\pi_1=(z,y,\gamma^{(1)})$.
Let $(\pi_2,\sigma_2)\in\mathcal{A}\left(\mathbf T_2,v(T-\mathbf T_2),z,\frac{y v(T-\mathbf T_1)}{v(T-\mathbf T_2)}\right)$
be a hedge such that $\pi_2=\left(z,y ,\gamma^{(2)}\right )$ is given by
$\gamma^{(2)}_t=\mathbb{I}_{t\leq \mathbf T_1}\frac{\gamma^{(1)}_t v(T-\mathbf T_1)}{v(T-\mathbf T_2)}, \ \ t\in [0,\mathbf T_2],$
and
$\sigma_2=\sigma_1\mathbb{I}_{\sigma_1<\mathbf T_1}+\mathbf T_2\mathbf{I}_{\sigma_1=\mathbf T_1}.$
Namely, the portfolio $\gamma^{(2)}$ is proportional $\gamma^{(1)}$ until
the moment $\mathbf T_1$ and then we sell the stocks. The stopping time
 $\sigma_2$ is almost the same as $\sigma_1$ with a small modification such that
 if $\sigma_1$ is equal to $\mathbf T_1$ then $\sigma_2=\mathbf T_2.$
From (\ref{2.4}) it follows that $(\pi_2,\sigma_2)\in \mathcal{A}\left(\mathbf T_2,v(T-\mathbf T_2),z,\frac{y v(T-\mathbf T_1)}{v(T-\mathbf T_2)}\right)$ and
 $V^{\pi_2}_{\sigma_2\wedge t}=V^{\pi_1}_{\sigma_1\wedge t}$, $t\in [0,\mathbf T_2].$

Let $\tau\in\mathcal{T}_{\mathbf T_2}$. Observe that if $\sigma_2<\tau$ then $\sigma_1<\tau\wedge\mathbf T_1$. Thus from (\ref{3.200}),
\begin{eqnarray*}
&R\left(\mathbf T_2,v,z,\frac{y v(T-\mathbf T_1)}{v(T-\mathbf T_2)}\right)\leq \cR(\mathbf T_2,v,\pi_2,\sigma_2)=\\
&\sup_{\tau\in\mathcal{T}_{\mathbf T_2}}\mathbb{E}_{\mathbb P}
\left(X^{\mathbf T_2,v}_{\sigma_2}\mathbb{I}_{\sigma_2<\tau}+
Y^{\mathbf T_2,v}_{\tau}\mathbb{I}_{\tau\leq\sigma_2}-V^{\pi_2}_{\sigma_2\wedge\tau}\right)^{+}\leq \\
& \sup_{\tau\in\mathcal{T}_{\mathbf T_2}}\mathbb{E}_{\mathbb P}
\left(X^{\mathbf T_1,v}_{\sigma_1}\mathbb{I}_{\sigma_1<\tau\wedge\mathbf T_1}+
Y^{\mathbf T_1,v}_{\tau\wedge\mathbf T_1}\mathbb{I}_{\tau\wedge\mathbf T_1\leq\sigma_1}-V^{\pi_1}_{\sigma_1\wedge\tau}\right)^{+}+\\
&\mathbb{E}_{\mathbb P}\left(
\sup_{0\leq t\leq \mathbf T_2}\left(|Y^{\mathbf T_1,v}_{t\wedge\mathbf T_1}-Y^{\mathbf T_2,v}_{t}|+
|X^{\mathbf T_1,v}_{t\wedge\mathbf T_1}-X^{\mathbf T_2,v}_{t}|\right)\right)\leq\\
&\epsilon+ R(\mathbf T_1,v,z,y)+m_v(|\mathbf T_1-\mathbf T_2|),
  \end{eqnarray*}
where
$$m_v(\delta)=\sup_{|t_2-t_1|\leq \delta}\mathbb{E}_{\mathbb P}\left(
\sup_{0\leq t\leq T}\left(|Y^{ t_1,v}_{t\wedge t_1}-Y^{ t_2,v}_{t}|+
|X^{t_1,v}_{t\wedge t_1}-X^{ t_2,v}_{t}|\right)\right).$$

From (\ref{2.3}) and the fact that $F,G$ are continuous it follows that $m_v$ is indeed a modulus of continuity.
Since $\epsilon>0$ was arbitrary we get
$$R\left(\mathbf T_2,v,z,\frac{y v(T-\mathbf T_1)}{v(T-\mathbf T_2)}\right)-R(\mathbf T_1,v,z,y)\leq m_v(|\mathbf T_1-\mathbf T_2|).$$
Next, we prove
$${R}(\mathbf T_1,v,z,y)-R\left(\mathbf T_2,v,z,\frac{y v(T-\mathbf T_1)}{v(T-\mathbf T_2)}\right)\leq m_v(|\mathbf T_1-\mathbf T_2|).$$
Choose $\epsilon>0$
and a hedge $(\tilde\pi_2,\tilde\sigma_2)\in\mathcal{A}\left(\mathbf T_2,v(T-\mathbf T_2),z,\frac{y v(T-\mathbf T_1)}{v(T-\mathbf T_2)}\right)$
which satisfy
\begin{equation*}\label{3.201}
R\left(\mathbf T_2,v,z,\frac{y v(T-\mathbf T_1)}{v(T-\mathbf T_2)}\right)> \cR(\mathbf T_2,v,\tilde\pi_2,\tilde\sigma_2)-\epsilon.
\end{equation*}
Denote $\tilde\pi_2=\left(z,\frac{y v(T-\mathbf T_1)}{v(T-\mathbf T_2)},\tilde\gamma^{(2)}\right)$.
Let $(\tilde\pi_1,\tilde\sigma_1)\in\mathcal{A}\left(\mathbf T_1,v(T-\mathbf T_1),z,y\right)$ be a hedge such that
$\tilde\pi_1=(z,y,\tilde\gamma^{(1)})$ is given by
$\tilde\gamma^{(1)}_t=\frac{\gamma^{(2)}_t v(T-\mathbf T_2)}{v(T-\mathbf T_1)}$, $t\in [0,\mathbf T_1]$, and
$\tilde\sigma_1=\tilde\sigma_2\wedge\mathbf T_1.$
Namely, we take a multiple of the hedge $(\tilde\pi_2,\tilde\sigma_2)$ and restrict it to the interval $[0,\mathbf T_1]$.
Let $\tau\in\mathcal{T}_{\mathbf T_1}$ and observe that if $\sigma_2<\tau$ then $\sigma_1=\sigma_2<\tau$. Thus we obtain,
 \begin{eqnarray*}
&R(\mathbf T_1,v,z,y)\leq \cR(\mathbf T_1,v,\tilde\pi_1,\tilde\sigma_1)=\\
&\sup_{\tau\in\mathcal{T}_{\mathbf T_1}}\mathbb{E}_{\mathbb P}
\left(X^{\mathbf T_1,v}_{\tilde\sigma_1}\mathbb{I}_{\tilde\sigma_1<\tau}+
Y^{\mathbf T_1,v}_{\tau}\mathbb{I}_{\tau\leq\tilde\sigma_1}-V^{\tilde\pi_1}_{\tilde\sigma_1\wedge\tau}\right)^{+}\\
&\leq \sup_{\tau\in\mathcal{T}_{\mathbf T_1}}\mathbb{E}_{\mathbb P}
\left(X^{\mathbf T_2,v}_{\tilde\sigma_2}\mathbb{I}_{\tilde\sigma_2<\tau}+
Y^{\mathbf T_2,v}_{\tau}\mathbb{I}_{\tau\leq\tilde\sigma_2}-V^{\tilde\pi_2}_{\tilde\sigma_2\wedge\tau}\right)^{+}+\\
&\mathbb{E}_{\mathbb P}\left(
\sup_{0\leq t\leq \mathbf T_2}\left(|Y^{\mathbf T_1,v}_{t\wedge\mathbf T_1}-Y^{\mathbf T_2,v}_{t}|+
|X^{\mathbf T_1,v}_{t\wedge\mathbf T_1}-X^{\mathbf T_2,v}_{t}|\right)\right)\leq\\
&\epsilon+ R\left(\mathbf T_2,v,z,\frac{y v(T-\mathbf T_1)}{v(T-\mathbf T_2)}\right)+m_v(|\mathbf T_1-\mathbf T_2|),
 \end{eqnarray*}
and by taking $\epsilon\downarrow 0$ we complete the proof.
\end{proof}

Next, we obtain the following simple result.
\begin{lem}\label{lem3.stock}
Let $v\in C_{++}[0,T]$.
There exists a continuous function $\tilde m_{v}:\mathbb {R}_{+}\rightarrow\mathbb{R}_{+}$
with $\tilde m_{v}(0)=0$ such that
for any $\tilde v\in C_{++}[0,T]$, $\mathbf T\in [0,T]$,
$z\geq 0$
and $y\in \mathbb R$, we have
$$\left|{R}\left(\mathbf T,v,z,y\right)-R\left(\mathbf T,\tilde v,z,\frac{y v(T-\mathbf T)}{\tilde v(T-\mathbf T)} \right)\right|\leq \tilde m_{v}(||v-\tilde v||).$$
\end{lem}
\begin{proof}
Fix $\mathbf T\in [0,T]$, $\tilde v\in C_{++}[0,T]$,
$z\geq 0$
and $y\in \mathbb R$.
Let
$\pi_1=(z,y,\gamma)$ be an \textit{admissible} portfolio in the market with maturity date $\mathbf T$ and an initial stock price
$v(T-\mathbf T)$. Consider the portfolio
$\pi_2=\left(z,\frac{y v(T-\mathbf T)}{\tilde v(T-\mathbf T)},\frac{\gamma v(T-\mathbf T)}{\tilde v(T-\mathbf T)}\right)$ as
an \textit{admissible} portfolio in the market with maturity date $\mathbf T$ and an initial stock price
$\tilde v(T-\mathbf T)$. The map $\pi_1\rightarrow\pi_2$ is a bijection between the corresponding sets of portfolios.
From (\ref{2.4}) it follows that
$V^{\pi_1}=V^{\pi_2}$, and so
\begin{eqnarray*}
&\left|{R}\left(\mathbf T,v,z,y\right)-R\left(\mathbf T,\tilde v,z,\frac{y v(T-\mathbf T)}{\tilde v(T-\mathbf T)} \right)\right|\leq\\
& \mathbb{E}_{\mathbb P}\bigg(
\sup_{0\leq t\leq \mathbf T}\left(|Y^{\mathbf T,v}_{t}-Y^{\mathbf T,\tilde v}_{t}|+
|X^{\mathbf T,v}_{t}-X^{\mathbf T,\tilde v}_{t}|\right)\bigg)\leq \tilde m_{v}(||v-\tilde v||)
\end{eqnarray*}
where
$$\tilde m_{v}(\delta)=\sup_{\mathbf T\in [0,T]}\sup_{||\tilde v-v||\leq\delta}\mathbb{E}_{\mathbb P}\left(
||X^{\mathbf T,v}-X^{\mathbf T,\tilde v}||+||Y^{\mathbf T,v}-Y^{\mathbf T,\tilde v}||
\right).$$
From (\ref{2.3}) and the fact that $F,G$ are continuous it follows that $\tilde m_{v}$ is indeed a modulus of continuity.
\end{proof}
Next, we establish continuity properties of the shortfall risk as a function of the initial position.
\begin{lem}\label{lem3.2}
Let $K>0$. There exists a constant $\hat C=\hat C(K)$ such that the following holds.
For any $v\in C_{++}[0,T]$,
$\mathbf T\in [0,T]$,  $(z_i,y_i)\in\mathbb{R}_{+}\times\mathbb{R}$, $i=1,2$, such that
$||v||\leq K$ and $y_1=0 \Leftrightarrow y_2=0$, we have
$$|R(\mathbf T,v,z_1,y_1)-R(\mathbf T,v,z_2,y_2)|\leq \hat C(|z_1-z_2|+|y_1-y_2|)^{3/4}.$$
\end{lem}
\begin{proof}
Fix $v\in C_{++}[0,T]$, $\mathbf T\in [0,T]$ and $(z_i,y_i)\in\mathbb{R}_{+}\times\mathbb{R}$, $i=1,2$. Assume that
$||v||\leq K$ and $y_1=0 \Leftrightarrow y_2=0$.
Denote $s=v(T-\mathbf T)$ and
$A=|z_1-z_2|+|y_1-y_2|$. Without loss of generality we assume that $0<A<1$.
Let $(\pi_1,\sigma_1)\in\mathcal{A}(\mathbf T,s,z_1,y_1)$ be
such that
\begin{equation}\label{3.2}
\cR(\mathbf T,v,\pi_1,\sigma_1)< R(\mathbf T,v,z_1,y_1)+A.
\end{equation}
Denote $\pi_1=(z,y,\gamma^{(1)})$. Define the stopping times
$\varrho=\inf\{t:\gamma^{(1)}_{t+}=0\}\wedge\mathbf T$
and
$$\varsigma=\inf\left\{t: V^{\pi_1}_t\leq  |z_2-z_1|+|y_2-y_1| (1+\mu)(s+{S}^{(s)}_{t\wedge\varrho}\right\}\wedge\mathbf T.$$
Introduce the hedge $(\pi_2,\sigma_2)\in\mathcal{A}(\mathbf T,s,z_2,y_2)$
by the relation $\sigma_2=\sigma_1$ and
$$\gamma^{(2)}_t=\mathbb{I}_{t\leq \varsigma}\left((\gamma^{(1)}_t+y_2-y_1)\mathbb{I}_{t\leq\varrho}+\gamma^{(1)}_t\mathbb{I}_{t>\varrho}\right). $$
Namely, at the time $\varsigma$ the investor liquidates the portfolio.
Until the time $\varsigma$,
the portfolio strategy $\gamma^{(2)}$ is a shift of $\gamma^{(1)}$ until the first stock liquidation time $\varrho$ of $\gamma^{(1)}$
and after this time the portfolios are the same.
From (\ref{2.4}) it follows
that
$$V^{\pi_2}_t=V^{\pi_2}_{t\wedge\varsigma}\geq V^{\pi_1}_{t\wedge\varsigma}-\left(|z_2-z_1|+ |y_2-y_1|(1+\mu)
(s+{S}^{(s)}_{t\wedge\varsigma\wedge\varrho})\right)\geq 0$$
where the last inequality follows from the definition of $\varsigma$.
Thus $(\pi_2,\sigma_2)\in\mathcal{A}(\mathbf T,s,z_2,y_2)$. Furthermore for any
random variable $\Phi$ we have the following
\begin{eqnarray}\label{3.3}
&(\Phi-V^{\pi_1}_t)^{+}\geq (\Phi-V^{\pi_2}_t)^{+}-\mathbb{I}_{t\leq\varsigma}
\left(|z_2-z_1|+ |y_2-y_1|(1+\mu) (s+{S}^{(s)}_{t\wedge\varsigma\wedge\varrho})\right)\\
&-\mathbb{I}_{t>\varsigma}\left(\mathbb{I}_{V^{\pi}_t>1}\Phi+\mathbb{I}_{V^{\pi}_t\leq 1}V^{\pi_1}_t\right).\nonumber
\end{eqnarray}
From (\ref{3.2})--(\ref{3.3}) we get
\begin{eqnarray*}
&R(\mathbf T,v,z_2,y_2)-R(\mathbf T,v,z_1,y_1)\leq \\
&A+ \cR(\mathbf T,v,\pi_2,\sigma_2)-\cR(\mathbf T,v,\pi_1,\sigma_1)\leq
A+|z_1-z_2|+\\
&\sup_{\tau\in\mathcal{T}_{\mathbf T}}\mathbb E_{\mathbb P}
\left(|y_2-y_1|(1+\mu)(s+{S}^{(s)}_{\tau\wedge\varsigma\wedge\varrho})+
\mathbb{I}_{\varsigma<\tau}\mathbb{I}_{V^{\pi_1}_{\tau}>1}X^{\mathbf T,v}_{\tau}+
\mathbb{I}_{\varsigma<\tau}\mathbb{I}_{V^{\pi_1}_{\tau}\leq 1}V^{\pi_1}_{\tau}
\right)\nonumber\\
&\leq A+|z_1-z_2|+ 2|y_2-y_1|(1+\mu)\mathbb E_{\mathbb P}\left(\max_{0\leq t\leq T} S^{(s)}_t\right)\\
&+\mathbb E_{\mathbb P}\left(\mathbb{I}_{\varsigma<\tau}\mathbb{I}_{V^{\pi_1}_{\tau}>1}X^{\mathbf T,v}_{\tau}+
\mathbb{I}_{\varsigma<\tau}\mathbb{I}_{V^{\pi_1}_{\tau}\leq 1}V^{\pi_1}_{\tau}\right).
\nonumber
\end{eqnarray*}
Thus in order to complete the proof it remains to show that
$$\sup_{\tau\in\mathcal{T}_{\mathbf T}}\mathbb E_{\mathbb P}\left(\mathbb{I}_{\varsigma<\tau}\mathbb{I}_{V^{\pi_1}_{\tau}>1}X^{\mathbf T,v}_{\tau}+
\mathbb{I}_{\varsigma<\tau}\mathbb{I}_{V^{\pi_1}_{\tau}\leq 1}V^{\pi_1}_{\tau}\right)=O(A^{3/4}).
$$
Let $\tau\in\mathcal{T}_T$. From
(\ref{2.martingale})--(\ref{2.3}) and the H\"older inequality (for $p=4$, $q=4/3$) it follows that
there exists a constant $c=c(K)$ such that
\begin{eqnarray*}
&\mathbb E_{\mathbb P}\left(\mathbb{I}_{\varsigma<\tau}\mathbb{I}_{V^{\pi_1}_{\tau}>1}X^{\mathbf T,v}_{\tau}+
\mathbb{I}_{\varsigma<\tau}\mathbb{I}_{V^{\pi_1}_{\tau}\leq 1}V^{\pi_1}_{\tau}\right)=\\
&\mathbb E_{\mathbb Q}\left(Z^{-1}_T\mathbb{I}_{\varsigma<\tau}\mathbb{I}_{V^{\pi_1}_{\tau}>1}X^{\mathbf T,v}_{\tau}+
Z^{-1}_T\mathbb{I}_{\varsigma<\tau}\mathbb{I}_{V^{\pi_1}_{\tau}\leq 1}V^{\pi_1}_{\tau}\right)\leq\\
&c\left((\mathbb E_{\mathbb Q} [\mathbb{I}_{\varsigma<\tau}\mathbb{I}_{V^{\pi_1}_{\tau}>1}])^{3/4}+
(\mathbb E_{\mathbb Q} [\mathbb{I}_{\varsigma<\tau}\mathbb{I}_{V^{\pi_1}_{\tau}<1}(V^{\pi_1}_{\tau})^{4/3}])^{3/4}\right)\leq\\
&c\left((\mathbb E_{\mathbb Q} [\mathbb{I}_{\varsigma<\tau}V^{\pi_1}_{\tau}])^{3/4}+
(\mathbb E_{\mathbb Q} [\mathbb{I}_{\varsigma<\tau}V^{\pi_1}_{\tau}])^{3/4}\right)\leq
\end{eqnarray*}
since ${\{V^{\pi_1}_t\}}_{t=0}^{\mathbf T}$ is a super--martingale with respect to $\mathbb Q$,
\begin{equation*}
\leq 2c(\mathbb E_{\mathbb Q} [\mathbb{I}_{\varsigma<\tau}V^{\pi_1}_{\varsigma+}])^{3/4}\leq
\end{equation*}
since
 $V^{\pi_1}_{\varsigma+}\leq  |z_2-z_1|+|y_2-y_1| (1+\mu)(s+{S}^{(s)}_{\varsigma\wedge\varrho})$
 on the event $\varsigma<\tau$,
\begin{equation*}
\leq 2c \left(|z_2-z_1|+2 s |y_2-y_1|(1+\mu)\right)^{3/4}=O(A^{3/4}),
\end{equation*}
and the proof is completed.
\end{proof}
\begin{lem}\label{lem3.10}
Let $K>0$, $v\in C_{++}[0,T]$ and $y\in\mathbb R$. Assume that $||v||\leq K$ and
$|y|\leq \frac{\delta}{\mu} \min_{t\in [0,T]}\frac{1}{v(t)}$. Then for any $z\geq 0$ and
$\mathbf T\in [0,T]$,
$$R(\mathbf T,v,z+\delta,0)-R(\mathbf T,v,z,y)\leq \hat C(K)|y|^{3/4}.$$
\end{lem}
\begin{proof}
Denote $s=v(T-\mathbf T)$. Without loss of generality we assume that $y\neq 0$.
Let $(\pi_1,\sigma_1)\in\mathcal{A}(\mathbf T,s,z,y)$
such that
$R(\mathbf T,v,z,y)> \cR(\mathbf T,v,\pi_1,\sigma_1)-|y|.$
Set $\pi_1=(z,y,\gamma^{(1)})$. Define the stopping times
$\varrho=\inf\{t:\gamma^{(1)}_{t+}=0\}\wedge\mathbf T$
and
$$\varsigma=\inf\left\{t: V^{\pi_1}_t\leq  |y| (1+\mu){S}^{(s)}_{t\wedge\varrho}\right\}\wedge\mathbf T.$$
Introduce the hedge $(\pi_2,\sigma_2)\in\mathcal{A}(\mathbf T,s,z+\delta,0)$
by the relation $\sigma_2=\sigma_1$ and $\pi_2=(z+\delta,0,\gamma^{(2)})$ where
$$\gamma^{(2)}_t=\mathbb{I}_{t\leq \varsigma}\left((\gamma^{(1)}_t-y)\mathbb{I}_{t\leq\varrho}+\gamma^{(1)}_t\mathbb{I}_{t>\varrho}\right). $$
Our assumptions imply that $g(y,s)=\delta$, and so
from (\ref{2.4}) we get
$$V^{\pi_2}_t=V^{\pi_2}_{t\wedge\varsigma}\geq V^{\pi_1}_{t\wedge\varsigma}-|y|(1+\mu) {S}^{(s)}_{t\wedge\varsigma\wedge\varrho}\geq 0.$$
Thus, similarly to Lemma \ref{lem3.2} we get
$$R(\mathbf T,v,z+\delta,0)-R(\mathbf T,v,z,y)\leq \hat C(K) |y|^{3/4},$$
as required.
\end{proof}

The following Proposition is the main result of this section.
\begin{prop}\label{essential}
(i). The function $R:[0,T]\times C_{++}[0,T]\times\mathbb{R}_{+}\times\mathbb{R}\rightarrow\mathbb{R}$
is upper semi continuous (and hence measurable). If we restrict the function $R$ to the domain
$[0,T]\times C_{++}[0,T]\times\mathbb{R}_{+}\times\mathbb{R}\setminus\{0\}$ then the function
$R:[0,T]\times C_{++}[0,T]\times\mathbb{R}_{+}\times\mathbb{R}\setminus\{0\}\rightarrow\mathbb R$
is a continuous function.\\
(ii). There exists a measurable function
$\beta^{*}:[0,T]\times C_{++}[0,T]\times\mathbb{R}_{+}\times\mathbb{R}\rightarrow\mathbb{R}$
such that $\beta^{*}(\mathbf T,v,\cdot,\cdot \cdot)$ depends only
on $v_{[0,T-\mathbf T]}$ and the infimum in
$\beta$ in (\ref{4.new+}) is attained at
$\beta^{*}=\beta^{*}(\mathbf T,v,z,y)$.\\
(iii).
For any $y\in\mathbb R$
the function $\hat{R}_y:[0,T]\times C_{++}[0,T]\times[\delta\mathbb{I}_{y=0},\infty)\rightarrow\mathbb{R}$
defined by
$$\hat R_y(\mathbf T,v,z)=\hat  R(\mathbf T,v,z,y)$$
is a continuous function.
\end{prop}
\begin{proof}
(i). Let $(\mathbf T,v,z,y)\in [0,T]\times C_{++}[0,T]\times\mathbb{R}_{+}\times\mathbb R$.
Clearly for any $(\tilde{\mathbf T},\tilde v,\tilde z,\tilde y)\in [0,T]\times C_{++}[0,T]\times\mathbb{R}_{+}\times\mathbb R$
we have
 \begin{eqnarray}\label{3.500}
&|R(\mathbf T,v,z,y)-R(\tilde{\mathbf T},\tilde v,\tilde z,\tilde y)|\leq\\
&\left|R(\mathbf T,v,z,y)-R\left(\tilde{\mathbf T},v,z,\frac{y v(T-{\mathbf T})}{v(T-\tilde{\mathbf T})}\right)\right|+\nonumber\\
&\left|R\left(\tilde{\mathbf T},v,z,\frac{y v(T-{\mathbf T})}{v(T-\tilde{\mathbf T})}\right)-R\left(\tilde{\mathbf T},\tilde v,
z,\frac{y v(T-{\mathbf T})}{\tilde v(T-\tilde{\mathbf T})}\right)\right|+\nonumber\\
&\left|R\left(\tilde{\mathbf T},\tilde v,
z,\frac{y v(T-{\mathbf T})}{\tilde v(T-\tilde{\mathbf T})}\right)-R(\tilde{\mathbf T},\tilde v,\tilde z,\tilde y)\right|.\nonumber
\end{eqnarray}
This together with Lemmas \ref{lem3.time}, \ref{lem3.stock},
and \ref{lem3.2}
yields that $R:[0,T]\times C_{++}[0,T]\times\mathbb{R}_{+}\times\mathbb{R}\setminus\{0\}\rightarrow\mathbb R$
is a continuous function. Next, we prove that
$R:[0,T]\times C_{++}[0,T]\times\mathbb{R}_{+}\times\mathbb{R}\rightarrow\mathbb R$ is upper semi continuous.
Set $s=v(T-\mathbf T)$.
For any \textit{admissible} portfolio
$\pi_1=(z,0,\gamma^{(1)})$ introduce the portfolio
$\pi_2=(z,y,\gamma^{(2)})$ by
$\gamma^{(2)}_0=y$ and $\gamma^{(2)}_t=\gamma^{(1)}_t$ for $t>0$.
Observe that for any $t>0$,
$$V^{\pi_2}_t-V^{\pi_1}_t=V^{\pi_2}_{0+}-V^{\pi_1}_{0+}=g(y,s)-g(y-\gamma^{(1)}_{0+},s)+g(\gamma^{(1)}_{0+},s)
\geq 0.$$
Thus $R(\mathbf T,v,z,y)\leq R(\mathbf T,v,z,0)$. This together with (\ref{3.500}) and
Lemmas \ref{lem3.time}, \ref{lem3.stock} and (\ref{lem3.2}) completes the proof.\\
(ii). Fix $(\mathbf T,v,z,y)\in [0,T]\times C_{++}[0,T]\times\mathbb{R}_{+}\times\mathbb R$.
Set $s=v(T-\mathbf T)$.
Assume that $\Gamma(s,z,y)\neq \emptyset$. We want to show that the minimum in
(\ref{4.new+}) is attained.
Thus, let $\{\beta_n\}_{n=1}^\infty\subset\Gamma(s,z,y)$ for which
\begin{equation}\label{3.502-}
\lim_{n\rightarrow\infty}R\left(\mathbf T,v,h(s,z,y,\beta_n),y+\beta_n\right)
=\inf_{\beta\in \Gamma(s,z,y)}R\left(\mathbf T,v,h(s,z,y,\beta),y+\beta\right).
\end{equation}
The set  $\Gamma(s,z,y)$ is compact, and so without loss of generality we assume (by passing to a sub sequence)
that the sequence $\{\beta_n\}_{n=1}^\infty\subset\Gamma(s,z,y)$ converges, thus let
$\lim_{n\rightarrow\infty}\beta_n=\hat\beta\in\Gamma(s,z,y)$. First assume that $\hat\beta\neq-y$, then
$$h(s,z,y,\hat\beta)=\lim_{n\rightarrow\infty} h(s,z,y,\beta_n).$$
Thus from (\ref{3.502-})
and the fact that
$R:[0,T]\times C_{++}[0,T]\times\mathbb{R}_{+}\times\mathbb{R}\setminus\{0\}\rightarrow\mathbb R$ is continuous we conclude that
\begin{equation}\label{3.502}
R\left(\mathbf T,v,h(s,z,y,\hat\beta),y+\hat\beta\right)
=\inf_{\beta\in \Gamma(s,z,y)}R\left(\mathbf T,v,h(s,z,y,\beta),y+\beta\right).
\end{equation}
Next, we deal with the case
$\hat\beta=-y$. In this case
$h(s,z,y,\hat\beta)=\delta+\lim_{n\rightarrow\infty} h(s,z,y,\beta_n)$, and so
from (\ref{3.502-}) and Lemma \ref{lem3.10} we obtain
\begin{equation}\label{3.503}
R\left(\mathbf T,v, h(s,z,y,\hat\beta),0\right)=
\inf_{\beta\in \Gamma(s,z,y)}R\left(\mathbf T,v,h(s,z,y,\beta),y+\beta\right).
\end{equation}
We conclude that the infimum in (\ref{4.new+}) is attained at $\hat\beta\in \Gamma(s,z,y)$. It follows that there exists a
measurable map $\beta^{*}:[0,T]\times C_{++}[0,T]\times\mathbb{R}_{+}\times\mathbb{R}\rightarrow\mathbb{R}$
such that $\beta^{*}=\beta^{*}(\mathbf T,v,z,y)$ depends only
on $v_{[0,T-\mathbf T]}$ and the infimum in (\ref{4.new+}) is attained at $\beta^{*}$
 provided that $\Gamma(v(T-\mathbf T),z,y)\neq \emptyset$.
For instance,
$\beta^{*}(\mathbf T,v,z,y)=0$ if $y=0$ and $z<\delta$ (i.e. $\Gamma\left(v(T-\mathbf T),z,y\right)=\emptyset$),
and if $\Gamma\left(v(T-\mathbf T),z,y\right)\neq\emptyset$,
\begin{eqnarray*}
&\beta^{*}(\mathbf T,v,z,y)=\min_{\tilde\beta\in\Gamma\left(v(T-\mathbf T),z,y\right)}
R\left(\mathbf T,v,h(s,z,y,\tilde\beta),y+\tilde\beta\right)=\\
&\inf_{\beta\in \Gamma(s,z,y)}R\left(\mathbf T,v,h(s,z,y,\beta),y+\beta\right).
\end{eqnarray*}
(iii).
Fix $y\in\mathbb {R}$. Choose a sequence $\{\mathbf T_n,v_n,z_n\}_{n=1}^\infty\subset [0,T]\times C[0,T]\times[\delta\mathbb{I}_{y=0},\infty)$
 which converges to $(\mathbf T,v,z)$.
 From the continuity of $G$ it follows that
 \begin{equation}\label{3.n}
 X^{\mathbf T,v}_0=\lim_{n\rightarrow\infty}  X^{\mathbf T_n,v_n}_0.
 \end{equation}
 Set $s=v(T-\mathbf T)$ and $s_n=v_n(T-\mathbf T_n)$, $n\in\mathbb N$.
Let $\beta_n=\beta^{*}(\mathbf T_n,v_n,z_n,y)$, $n\in\mathbb N$.
 The sequence $\beta_n$, $n\in\mathbb {N}$ is bounded, and so without loss of generality by taking a subsequence
we can assume that it converges. Denote $\hat\beta=\lim_{n\rightarrow\infty}\beta_n$.
It is straightforward to check that $\hat\beta\in \Gamma(s,z,y)$.  By using (\ref{3.n}) and applying
similar arguments as in (\ref{3.502})--(\ref{3.503}) we get
\begin{eqnarray*}
&\hat R(\mathbf T,v,z,y)\leq \min\left((X^{\mathbf T,v}_0-z)^{+},
R\left(\mathbf T,v,h(s,z,y,\hat\beta),\hat\beta+y\right)\right)\\
&\leq
\lim\inf_{n\rightarrow\infty}\hat R(\mathbf T_n,v_n,z_n,y).
\end{eqnarray*}
This yields the lower semi--continuity of $\hat R_y$. Thus, it remains to establish upper semi-continuity.
Let $\tilde\beta=\beta^{*}(\mathbf T,s,z,y)$. First assume that
$h(s,z,y,\tilde\beta)>0$.
Then for sufficiently large $n$, we have $\tilde\beta\in\Gamma(s_n,z_n,y)$ and
$$h(s,z,y,\tilde\beta)=\lim_{n\rightarrow\infty} h(s_n,z_n,y,\tilde\beta).$$
Thus from (\ref{3.n}) and the fact that $R$ is upper semi--continuous we get
\begin{eqnarray}\label{3.505}
&\hat R(\mathbf T,v,z,y)= \min\left((X^{\mathbf T,v}_0-z)^{+},
R\left(\mathbf T,v,h(s,z,y,\tilde\beta),\tilde\beta+y\right)\right)\\
&\geq
\lim\sup_{n\rightarrow\infty}\hat R(\mathbf T_n,v_n,z_n,y).\nonumber
\end{eqnarray}
Finally, assume that $h(s,z,y,\tilde\beta)=0$.
Let $\pi=(0,\tilde y,\gamma)$ be an \textit{admissible} portfolio
for some $\tilde y$. From the fact that the geometric Brownian motion can increase or decrease for any amount (with positive probability)
on any time interval it follows that
 $\gamma_t=0$ for $t>0$. Indeed, for otherwise the portfolio value can become negative with positive probability. Thus
$$R(\mathbf T,v,0,\tilde y)=\inf_{\sigma\in\mathcal{T}_{\mathbf T}}\sup_{\tau\in\mathcal{T}_{\mathbf T}}
\mathbb{E}_{\mathbb P}\left(X^{\mathbf T,v}_{\sigma}\mathbb{I}_{\sigma<\tau}+
Y^{\mathbf T,v}_{\tau}\mathbb{I}_{\tau\leq\sigma}\right).$$
In particular, $R(\mathbf T,v,0,\tilde y)$ does not depend on $\tilde y$.
From (\ref{2.3}) and the fact that $F,G$ are continuous it follows that
$R(\cdot,\cdot\cdot,0,\tilde y)$ is continuous.
For any $n\in\mathbb {N}$, $-y\in\Gamma(s_n,z_n,y)$, and so from the upper semi continuity of $R$ it follows that
\begin{eqnarray}\label{3.506}
&\hat R(\mathbf T,v,z,y)= \min\left((X^{\mathbf T,v}_0-z)^{+},R (\mathbf T,v,0,\tilde\beta+y)\right)\\
&\geq\limsup_{n\rightarrow\infty}
\min\left((X^{\mathbf T_n,v_n}_0-z_n)^{+},R (\mathbf T_n,v_n,z_n-\delta\mathbb{I}_{y=0},0)\right)\nonumber\\
&\geq\lim\sup_{n\rightarrow\infty}\hat R(\mathbf T_n,v_n,z_n,y).
\nonumber
\end{eqnarray}
From (\ref{3.505})--(\ref{3.506}) we derive the upper semi continuity of $\hat R_y$, and the proof is completed.
\end{proof}

\section{Acknowledgement}
The authors were partially supported by Einstein Foundation Grant no.A 2012 137. The first author also
acknowledges support of the Marie Curie Actions fellowships Grant no.618235 and the second author of Israeli Science
Foundation Grant no.82/10.


\begin{thebibliography}{99}
\footnotesize



\bibitem{AJS}
{\sc Altarovici, A., Muhle-Karbe, J.
and Soner, H.M.} (2015).
Asymptotics with fixed transaction costs,
{\em Finance and Stochastics.} {\bf 19,} 363--414.


\bibitem{CK}
{\sc Cvitanic, J. and Karatzas, I.} (1996).
Backward stochastic differential equations with reflection and Dynkin games,
{\em Annals of  Probability.} {\bf 24,} 2024-–2056.

\bibitem{D}
{\sc  Dolinsky, Y.}
Limit Theorems for partial hedging under transaction Costs. (2014).
{\em Math. Finance.} {\bf 24,} 567--597.


\bibitem{D1} {\sc  Dolinsky, Y.} (2013).
Hedging of game Options with the presence of transaction costs.
{\em Ann. Appl. Probab.} {\bf 23,}
2212--2237.


\bibitem{DK1} {\sc Dolinsky, Y. and Kifer, Y.} (2007). Hedging with risk for
game options in discrete time. {\em Stochastics.}
{\bf 79,}
169--195.

\bibitem{DK2}
{\sc Dolinsky, Y. and Kifer, Y.} (2008). Binomial approximations of
shortfall risk for game options. {\em Ann. Appl. Probab.} {\bf 18,}
1737--1770.


\bibitem{EH}
{\sc Eastham, J.F. and Hastings, K.J.} (1988).
Optimal impulse control of portfolios.
{\em Math.Oper.Res.} {\bf 13}, 588--605.


\bibitem{G1}
{\sc Guasoni, P.} (2002).
Risk minimization under transaction costs.
{\em Finance and Stochastics.} {\bf 6}, 91--113.

\bibitem{G2}
{\sc Guasoni, P.} (2002).
Optimal investment with transaction costs and without semimartingales.
{\em Ann.Appl.Probab.} {\bf 12}, 1227--1246.


\bibitem{IK}
{\sc Iron, Y. and Kifer, Y.}
Hedging of swing game options in continuous time.
{\em Stochastics.} {\bf 83}, 365--404.


\bibitem{Ka1}
{\sc Kamizono, K.} (2001).
Partial hedging under proportional transaction costs.
{\em PhD dissertation.}

\bibitem{Ka2}
{\sc Kamizono, K.} (2003).
Partial hedging under transaction costs.
{\em SIAM J.Control Optimization.} {\bf 5,} 1545--1558.

\bibitem{K}
{\sc Kifer, Y.} (2006).
Error estimates for binomial approximations of game options.
{\em Ann.Appl.Probab.} {\bf 16}, 984--1033.

\bibitem{Ki1}
{\sc Kifer, Y.} (2000). Game options. {\em Finance and Stoch.} {\bf 4,}
 443--463.

\bibitem{Ki2}
{\sc Kifer, Y.} (2013).
Dynkin games and Israeli options.
{\em ISRN Probability and Statistics.} Id.856458.


\bibitem{Ko}
{\sc Korn, R.} (1998).
Portfolio optimisation with strictly positive transaction costs and impulse control.
{\em Finance and Stochastics.}{\bf 2}, 85--114.

\bibitem{KQ}
{\sc Kobylanski, M. and Quenez, M.C.} (2012).
Optimal stopping time problem in a general
framework.
{\em Electron. J. Probab.} {\bf 17}, 1--28.


\bibitem{LM}
{\sc Lepeltier, J.P. and Maingueneau, J.P.} (1984).
Le jeu de Dynkin en theorie generale sans l’hypothese de Mokobodski.
{\em Stochastics.} {\bf 13}, 24-–44.

\bibitem{LMW}
{\sc Lo, A.W., Mamaysky, H. and Wang, J.} (2004).
Asset prices and trading volume under fixed transaction costs.
{\em J. Polit. Econ.} {\bf 112}, 1054-–1090.


\bibitem{MP}
{\sc Morton, A.J. and Pliska, S.R.} (1995).
Optimal portfolio management with fixed transaction costs.
{\em Math. Finance} {\bf 5}, 337--356.


\bibitem{OP} B.Oksendal and A.Sulem,
{\sc Oksendal, B. and Sulem, A.} (2002).
Optimal consumption and portfolio with both fixed and proportional transaction costs‏.
{\em SIAM J.Control Optimization.} {\bf 6,} 1765--1790.








\end{thebibliography}
\end{document}